\documentclass[reqno,a4paper,11pt]{amsart}

\usepackage{amsmath}
\usepackage{amsfonts}
\usepackage{amssymb}
\usepackage{mathrsfs}
\usepackage[pdfborder={0 0 0}]{hyperref}
\usepackage{color}
\usepackage{paralist}

\newtheorem{theorem}{Theorem}[section]
\newtheorem{proposition}[theorem]{Proposition}
\newtheorem{lemma}[theorem]{Lemma}

\newtheorem{definition}[theorem]{Definition}

\numberwithin{equation}{section}

\newenvironment{propositionlist}{\begin{compactenum}[\itshape i)]}{\end{compactenum}}
\newenvironment{definitionlist}{\begin{compactenum}[\itshape i)]}{\end{compactenum}}


\newcommand{\refitem}[1] {\textit{\ref{#1})}}

\newlength{\dinwidth}
\newlength{\dinmargin}

\usepackage{oldgerm}
\usepackage{color}
\usepackage{bbm}
\usepackage{amssymb}
\usepackage{epsfig}
\usepackage{amsthm}
\usepackage{amsmath, mathrsfs,psfrag}

\renewcommand{\mathbb}[1]{\mathbbm{#1}}

\newcommand{\Cl}{\mathbbm{C}}
\newcommand{\Rl}{\mathbb{R}}
\newcommand{\Nl}{\mathbb{N}}

\definecolor{lightgray}{rgb}{0.8,0.8,0.8}

\newcommand{\Om}{\Omega}

\newcommand{\te}{\theta}

\newcommand{\la}{\lambda}

\newcommand{\La}{\Lambda}
\newcommand{\eps}{\varepsilon}

\newcommand{\A}{\mathcal{A}}

\newcommand{\B}{\mathcal{B}}

\newcommand{\I}{\mathcal{I}}
\newcommand{\K}{\mathcal{K}}

\newcommand{\Hil}{\mathcal{H}}
\newcommand{\VV}{\mathcal{V}}

\newcommand{\Q}{\mathcal{Q}}
\newcommand{\F}{\mathcal{F}}
\newcommand{\U}{\mathcal{U}}
\newcommand{\DD}{\mathcal{D}}
\newcommand{\W}{\mathcal{W}}

\newcommand{\SF}{\mathcal{S}}

\newcommand{\Ss}{\mathscr{S}}   

\newcommand{\fti}{\tilde{f}}

\newcommand{\hti}{\tilde{h}}

\newcommand{\Wti}{\tilde{W}}

\newcommand{\fhat}{\hat{f}}

\newcommand{\hhat}{\hat{h}}

\newcommand{\PGp}{\mathcal{P}_+}   

\newcommand{\PGpo}{\mathcal{P}_+^\uparrow}   

\newcommand{\PG}{\mathcal{P}}

\newcommand{\frS}{\textfrak{S}}

\newcommand{\balpha}{{\boldsymbol{\alpha}}}
\newcommand{\bbeta}{{\boldsymbol{\beta}}}

\newcommand{\btau}{{\boldsymbol{\tau}}}

\def\bte{{\boldsymbol{\theta}}}

\newcommand{\OO}{O}

\newcommand{\supp}{\text{supp}\,}

\newcommand{\sh}{\mathrm{sh}}
\newcommand{\ch}{\mathrm{ch}}

\newcommand{\Strip}{\mathrm{S}}

\newcommand{\ot}{\otimes}

\newcommand{\zd}{z^{\dagger}}

\newcommand{\ad}{a^{\dagger}}

\title[Operator-algebraic construction of integrable gauge theories]{Towards an operator-algebraic construction of integrable global gauge theories}
\author{Gandalf Lechner}
\address{Department of Physics, Vienna University, 1090 Vienna, Austria (Present address: Institute for Theoretical Physics, Leipzig University, 04103 Leipzig, Germany)}
\email{gandalf.lechner@uni-leipzig.de}
\author{Christian Sch\"utzenhofer}
\address{Department of Physics, Vienna University, 1090 Vienna, Austria}
\email{schuetzenhofer@gmx.at}
\date{August 11, 2012}

\begin{document}
\maketitle

\begin{abstract}
	The recent construction of integrable quantum field theories on two-dimensional Minkowski space by operator-algebraic methods is extended to models with a richer particle spectrum, including finitely many massive particle species transforming under a global gauge group. Starting from a two-particle S-matrix satisfying the usual requirements (unitarity, Yang-Baxter equation, Poincar\'e and gauge invariance, crossing symmetry, ...), a pair of relatively wedge-local quantum fields is constructed which determines the field net of the model. Although the verification of the modular nuclearity condition as a criterion for the existence of local fields is not carried out in this paper, arguments are presented that suggest it holds in typical examples such as nonlinear $O(N)$ $\sigma$-models. It is also shown that for all models complying with this condition, the presented construction solves the inverse scattering problem by recovering the S-matrix from the model via Haag-Ruelle scattering theory, and a proof of
asymptotic completeness is given.
\end{abstract}

\section{Introdcution}

Completely integrable quantum field theories on two-dimensional Minkowski space have attracted the interest of physicists and mathematicians for a long time. On the one hand, such models are interesting in their own right, as they provide examples of non-trivial quantum field theories which are simple enough to be accessible to thorough analysis from many different points of view. On the other hand, some of these models resemble certain aspects of much more complicated systems of direct physical relevance. In particular, non-linear $\sigma$-models in two dimensions are believed to exhibit certain features of non-Abelian gauge theories in four dimensions, such as asymptotic freedom (see for example  \cite{AbdallaAbdallaRothe:1991}).

In comparison to quantum field theories in higher dimensions, the dynamics of integrable models are severely restricted by an infinity of conversation laws, which for example exclude particle production in scattering processes of any energy. Despite these simplifying features, a rigorous construction of integrable quantum field theories beyond perturbation theory is often a difficult task. In some cases, such as the Sine-Gordon model, a construction by the Euclidean methods of constructive quantum field theory is possible \cite{Frhlich:1975}, whereas in other cases, such as the $O(N)$ $\sigma$-models, the current state of the art is a construction by perturbative renormalization, with the usual problems of controlling the perturbation series (see \cite[Ch.~7]{AbdallaAbdallaRothe:1991} and the references cited there).

In fact, the Lagrangians of these models are more complicated then those of models with polynomial self-interaction, which are under complete control in two dimensions \cite{GlimmJaffe:1987}. On the other hand, the S-matrix of integrable quantum field theories has a very simple structure: No particle production occurs, processes with $n$ incoming and $n$ outgoing particles are described by products of two-particle S-matrices, and furthermore, the elastic two-particle S-matrix is constrained by the Yang-Baxter-relation and other conditions \cite{AbdallaAbdallaRothe:1991}. Thus for integrable models, the two-particle S-matrix instead of the Lagrangian seems to be a much more convenient object for describing the interaction. In particular, formulating the problem of constructing integrable quantum field theories as an inverse scattering problem starting from a given two-particle S-matrix sidesteps all problems related to quantization and renormalization.

This inverse scattering point of view lies at the heart of two different approaches to the construction of integrable models. In the form-factor program \cite{BergKarowskiWeisz:1979,Smirnov:1992,BabujianFoersterKarowski:2006}, the aim is to calculate $n$-point functions of local field operators, thus constructing the models in the Wightman framework \cite{StreaterWightman:1964} of quantum field theory. The basic object of interest here are the form factors, matrix elements of local field operators in scattering states, which for many models can be explicitly computed from the S-matrix and analyticity assumptions, see for example \cite{BabujianFoersterKarowski:2012} for recent results containing $O(N)$-symmetric models . The $n$-point functions are then given as infinite series of integrals over form factors. Although the convergence of this series is expected to be much better than the usual perturbation series, it is presently under control only in a few special cases \cite{BabujianFoersterKarowski:2006}.

The second, and much more recent, inverse scattering approach to integrable models makes use of the operator-algebraic framework of quantum field theory \cite{Haag:1996}. This program was initiated by Schroer's insight \cite{Schroer:1999} that the crossing symmetry of the S-matrix, mathematically similar to the KMS property for the vacuum state on an algebra of observables localized in the Rindler wedge $W_R:=\{(x_0,x_1)\in\Rl^2\,:\,x_1\geq|x_0|\}$ with respect to the Lorentz boost dynamics, allows for an explicit construction of quantum fields which are localized in $W_R$. These fields are important auxiliary objects in the construction, called polarization-free generators \cite{SchroerWiesbrock:2000-1} because of their simple momentum-space properties, see also \cite{BorchersBuchholzSchroer:2001, Mund:2010} for a model-independent analysis of this concept. For the case of a particle spectrum consisting of just one species of neutral, massive particles, a complete construction of these fields was carried
out in \cite{Lechner:2003}. One can then pass to the von Neumann algebras they generate and efficiently characterize all local field operators present in the model at hand by operator-algebraic techniques \cite{BuchholzLechner:2004}. Bypassing all problems related to the explicit construction of point-local field operators, existence of local fields can be proven with the help of the modular nuclearity condition of Buchholz, D'Antoni and Longo \cite{BuchholzDAntoniLongo:1990-1}.

Along these lines, an infinite family of integrable models (containing a single species of massive particles) were constructed, and it was shown that the construction yields quantum field theories which are asymptotically complete and solve the inverse scattering problem, i.e. the initially considered factorizing S-matrix can be recovered in scattering theory \cite{Lechner:2008}. Their short distance limits have been studied in \cite{BostelmannLechnerMorsella:2011}, and generalizations to higher dimensions in the context of deformation procedures can be found in \cite{Lechner:2011,Alazzawi:2012}.

It is the aim of the present article to generalize this construction to theories with a richer particle spectrum, containing an arbitrary number of massive particle species, which can also carry arbitrary charges and transform under some global gauge group. This more general class contains in particular the $O(N)$ $\sigma$-models. Whereas many of the basic ideas underlying this construction are the same as in the scalar case, the appearance of many particle species and a gauge group changes the structure of the S-matrix, and the construction has to be reconsidered. In this article, we will proceed as follows: In Section~2, we specify our precise assumptions on the single particle spectrum and the two-particle S-matrix. We then construct a convenient vacuum Hilbert space from these data. This Hilbert space carries a representation of two different versions of the Zamolodchikov-Faddeev algebra \cite{ZamolodchikovZamolodchikov:1979, Faddeev:1984}, and their relative exchange relations are calculated. These
Zamolodchikov creation/annihilation operators are then combined to a pair of quantum fields $\phi,\phi'$
in Section~3. The transition to the operator-algebraic setting is discussed in Section~4, and we also review the strategy for proving existence of local field operators there. The complete investigation of the modular nuclearity condition goes beyond the scope of the article, but we outline the basic strategy and also argue why this condition is likely to hold in the case of the $O(N)$ $\sigma$-models. In Section~5, we show that whenever a model complies with the modular nuclearity condition, our construction solves the inverse scattering problem and yields an asymptotically complete theory. Finally, in Section~6 we give a few explicit examples of $S$-matrices fitting into our framework, in particular the $O(N)$ $\sigma$-models.

This article is partly based on the diploma thesis of the second named author \cite{Schutzenhofer:2011}.

\section{Two-particle S-matrices for general particle spectra and $S$-symmetric Fock spaces}\label{section:S-Matrices}

The construction of the models we are interested in begins with the specification of their single particle mass and charge spectra. We thus consider a compact Lie group $G$ as the global gauge group, and identify charges with equivalence classes $q$ of unitary irreducible representations of $G$ as usual.

As charges carried by single particles, we consider a set $\Q$ of finitely many charges, and to account for antiparticles, we assume that with each class $q\in\Q$, also the conjugate class $\overline{q}$ is contained in $\Q$. We are interested in constructing massive stable quantum field theories and must therefore guarantee that in each sector, the masses are positive isolated eigenvalues of the mass operator. This will in particular be the case when to each charge $q$ there corresponds a single mass $m(q)>0$ (with $m(\overline{q})=m(q)$), and for simplicity, we restrict ourselves to this setting\footnote{At the cost of a little more notational effort, our results can be shown to also hold in the more general case where in each sector there exist finitely many masses $m(q)_k$ with mass shells separated from the rest of the energy-momentum spectrum in that sector.}.

Since we are working in two spacetime dimensions, states of a single particle of fixed mass $m>0$ and charge $q$ can be described in momentum space by square integrable rapidity wave functions in
$L^2(\Rl,d\te)$, where the rapidity $\te$ is connected to the on-shell momentum via
\begin{align}\label{eq:palpha}
	p_m(\te)
	:=
	m\left(
	\begin{array}{c}
	\cosh\te\\
	\sinh\te
	\end{array}
	\right)
	\,.
\end{align}
On $L^2(\Rl,d\te)$,  the proper orthochronous Poincar\'e group $\PGpo$ acts via the unitary, strongly continuous, positive energy, irreducible representations
\begin{align}
	(U_{1,m}(x,\la)\psi)(\te)
	:=
	e^{ip_m(\te)\cdot x}\cdot \psi(\te-\la)
	\,,\qquad m>0
	\,,
\end{align}
where $(x,\la)\in\PGpo$ denotes the Poincar\'e transformation consisting of a boost with rapidity $\la\in\Rl$ and a subsequent space-time translation by $x\in\Rl^2$.

For several particle species, the single particle Hilbert space has the form
\begin{align}\label{eq:OneParticleSpace}
	\Hil_1
	:=
	L^2(\Rl,d\te)\ot\K
	\,,
\end{align}
where $\K$ is a finite-dimensional Hilbert space, $D:=\dim\K$. More precisely, we decompose $\Hil_1$ into subspaces of fixed charge $q\in\Q$ and mass $m(q)$,
\begin{align}
	\Hil_1
	=
	\bigoplus_{q\in\Q}\Hil_{1,q}
	\,,\qquad
	\Hil_{1,q}=L^2(\Rl,d\te)\ot\K_q\,,
\end{align}
where the gauge group $G$ acts on $\K_q$ via a unitary irreducible representation $V_{1,q}$ in the class $q\in\Q$, and trivially on $L^2(\Rl,d\te)$, and the Poincar\'e group acts on $\Hil_{1,q}$ via the representation $U_{1,m(q)}\ot{\rm id}_{\K_q}$. Clearly the two group representations
\begin{align}\label{eq:U1}
	U_1
	&:=
	\bigoplus_{q\in\Q} \left(U_{1,m(q)}\ot{\rm id}_{\K_q}\right)
	\,,\qquad
	V_1
	:=
	\bigoplus_{q\in\Q} \left({\rm id}_{L^2(\Rl,d\te)}\ot V_{1,q}\right)
\end{align}
are unitary and commute.

Some examples of models with a single mass $m>0$ might help to illustrate this setting: {\em i)}~For a neutral particle, take $G=\{e\}$ and $\Hil_1=L^2(\Rl,d\te)$, {\em ii)} for a model of
electric charge, take $G=U(1)$ and the two conjugate irreducible representations $V_{1,\pm}(e^{i\vartheta})=e^{\pm i\vartheta}$ on $\K_\pm=\Cl$, {\em iii)}~for an $O(N)$ $\sigma$-model, take
$G=O(N)$ for some $N\geq3$, and the defining self conjugate irreducible representation of $O(N)$ on $\K=\Cl^N$. The special case {\em i)} of a single neutral particle species will be referred to as the
scalar case. All our subsequent analysis reduces to the known results established in \cite{Lechner:2003} and \cite{Lechner:2008} for the scalar case.
\\
\\
In the following, we will always tacitly refer to a fixed particle spectrum given by the data $G,\Q,\{V_{1,q}\}_{q\in\Q},\{m_q\}_{q\in\Q}$ and complying with the above assumptions. It will be convenient to use a particular orthonormal basis for $\K$ \eqref{eq:OneParticleSpace}: For each subspace $\K_q$ of fixed charge, we choose an orthonormal basis, and denote their direct sum by $\{e^\alpha\,:\,\alpha=1,...,D\}$. We can thus associate with each index $\alpha$ a definite charge $q_{[\alpha]}$ and mass $m_{[\alpha]}:=m(q_{[\alpha]})$. The corresponding components of vectors $\Psi_1\in\Hil_1$ will be denoted by $\te\mapsto\Psi_1^\alpha(\te)$.

We will write $(\,\cdot\,,\,\cdot\,)$ for the scalar product on $\K$, put $\I:=\{1,...,D\}$, and make use of standard multi index notation for tensor products. For example, we write
$\xi^\balpha:=(e^{\alpha_1}\ot...\ot e^{\alpha_n},\,\xi)$, $\balpha=(\alpha_1,...,\alpha_n)$, for vectors $\xi\in\K^{\ot n}$, and $T^\balpha_\bbeta:=(e^{\alpha_1}\ot...\ot
e^{\alpha_n},\,T\,e^{\beta_1}\ot...\ot e^{\beta_n})$ for tensors $T\in\B(\K^{\ot n})$, $n\in\Nl$. Furthermore, given
$T\in\B(\K\ot\K)$ and $n\geq2$, we will use the shorthand notation
$T_{n,k}:=1_{k-1}\ot T\ot 1_{n-k-1}$, $k=1,...,n-1$, where $1_r$ denotes the identity on $\K^{\ot r}$.

The description of the single particle structure is completed by a remark on the TCP symmetry. In view of our above assumption regarding conjugate charges $q,\overline{q}\in\Q$, we
have a single particle TCP operator $J_1$ on $\Hil_1$ (See for example \cite{DoplicherHaagRoberts:1974, BuchholzFredenhagen:1982, GuidoLongo:1995, Mund:2001}). It is the product of a charge conjugation operator exchanging the representation spaces $\K_q$ and $\K_{\overline{q}}$, and a space-time reflection, acting by complex conjugation on $L^2(\Rl,d\te)$. When working in the basis $e^\alpha$, this simply means that we have an
involution $\alpha\mapsto\overline{\alpha}$ of $\{1,...,D\}$ (that is, a permutation of $D$ elements with $\overline{\overline{\alpha}}=\alpha$) such that $m_{[\overline{\alpha}]}=m_{[\alpha]}$ and
$q_{[\overline{\alpha}]}=\overline{q_{[\alpha]}}$, and the TCP operator reads
\begin{align}\label{eq:J1}
	(J_1\Psi_1)^\alpha(\te)
	:=
	\overline{\Psi_1^{\overline{\alpha}}(\te)}
	\,.
\end{align}
By straightforward calculation, one checks that $J_1$ is an antiunitary involution which commutes with $V_1$ and extends the representation $U_1$ to the proper Poincar\'e group $\PGp$, including the space-time reflection
$j(x_0,x_1):=(-x_0,-x_1)$, by setting $U_1(j):=J_1$.
\\
\\
We now come to specifying the interaction of the models to be constructed. Our point of view is that of inverse scattering theory, and since we want to study completely integrable models, we take a
{\em factorizing} S-matrix as an input to our construction. Such a collision operator is completely fixed by collision processes with two incoming and two outgoing particles, and enjoys a number
of special properties \cite{Dorey:1998}:
\begin{itemize}
	\item There is no particle production.
	\item The S-matrix kernels for processes with $n$ incoming and $n$ outgoing particles are products of kernels of $(2\to2)$ processes.
	\item The sets of incoming and outgoing momenta coincide.
	\item Particles of different mass do not interact.
\end{itemize}
This particular structure makes the two particle S-matrix the main object of interest in factorized scattering. By Lorentz invariance, processes with incoming particles of types $\alpha,\beta$ and rapidities $\te_1,\te_2$, and outgoing types $\gamma,\delta$ and rapidities $\te_1',\te_2'$ are governed by an amplitude $S^{\alpha\beta}_{\gamma\delta}(\te)$ depending only on the difference of rapidities $\te=\te_1-\te_2=\te_1'-\te_2'$, i.e. we can view $S$ as a map $\Rl\to\B(\K\ot\K)$. Further constraints on $S$ arise from general S-matrix properties like unitarity, Poincar\'e invariance, TCP invariance, crossing symmetry and hermitian analyticity. For a thorough discussion of all these standard properties, we refer to the textbook and review \cite{AbdallaAbdallaRothe:1991, Mussardo:1992, Dorey:1998}. Note that there several different conventions regarding the positions of the indices on $S$ are used in the literature. For example, in many places one finds the order of the two upper indices reversed, i.e. $S^{\beta\alpha}_{\gamma\delta}$ instead of $S^{\alpha\beta}_{\gamma\delta}$.

We now give an abstract definition of the class of $S$-matrices we will consider, using the notation $\Strip(a,b):=\{\zeta\in\Cl\,:\,a<{\rm Im}\zeta<b\}$ for strips in the complex plane.

\begin{definition}\label{definition:SMatrix}
     An $S$-matrix (corresponding to the particle spectrum given by $G, \Q, \{V_{1,q}\}_{q\in\Q}$, $\{m_q\}_{q\in\Q}$) is a continuous bounded function $S:\overline{\Strip(0,\pi)}\to\B(\K\ot\K)$ which is analytic in the interior of this strip and satisfies for arbitrary $\te,\te'\in\Rl$, $\alpha,\beta,\gamma,\delta\in\I$,
     \begin{definitionlist}
	  \item\label{item:Unitarity} Unitarity:
	  \begin{align}
	       S(\te)^*=S(\te)^{-1}
	  \end{align}
	  \item\label{item:HermitianAnalyticity} Hermitian analyticity:
	  \begin{align}
	       S(\te)^{-1}=S(-\te)
	  \end{align}
	  \item\label{item:YangBaxterEqn} Yang-Baxter equation:
	  \begin{align}
	     (S(\te)\ot1_1)(1_1\ot S(\te+\te'))(S(\te')\ot 1_1)=(1_1\ot S(\te'))(S(\te+\te')\ot 1_1)(1_1\ot S(\te))
	  \end{align}
	  \item\label{item:TranslationInvariance} Translational invariance\footnote{It will become apparent later that this is the right condition for ensuring translational invariance of $S$. Also note that Lorentz invariance poses no further condition as $S$ will only depend on differences of rapidities.}:
	  \begin{align}
	       S^{\alpha\beta}_{\gamma\delta}(\te)=0\;\text{ if }\; m_{[\alpha]} \neq m_{[\delta]}\;\text{ or }\; m_{[\beta]}\neq m_{[\gamma]} .
	  \end{align}
	  \item\label{item:TCPInvariance} TCP invariance:
	  \begin{align}
	       S^{\alpha\beta}_{\gamma\delta}(\te)
	       =
	       S^{\overline{\delta}\overline{\gamma}}_{\overline{\beta}\overline{\alpha}}(\te)
	  \end{align}
	\item\label{item:GaugeInvariance} Gauge invariance:
	\begin{align}
		[S(\te),\,V_1(g)\ot V_1(g)]=0
		\,,\qquad
		g\in G,\;\te\in\Rl\,.
	\end{align}
	\item\label{item:Crossing} Crossing symmetry:
	\begin{align}
		S^{\alpha\beta}_{\gamma\delta}(i\pi-\te)
		=
		S^{\overline{\gamma}\alpha}_{\delta\overline\beta}(\te)
	\end{align}
     \end{definitionlist}
     The family of all S-matrices will be denoted $\SF$.
\end{definition}

In view of the required invariance properties of $S$, many of its components are related or have to vanish. For example, in the case of a theory with a single mass, gauge group $G=U(1)$, and the two conjugate representations $V_{1,\pm}(e^{i\vartheta})=e^{\pm i\vartheta}$, we have $D=2$ and hence $S(\te)$ can be viewed as a $(4\times4)$-matrix. But as a consequence of gauge invariance, TCP invariance, and crossing symmetry, only two of its 16 components are non-zero and independent. In many articles on factorizing S-matrices, these amplitudes are taken as the main quantities of interest. For our approach, however, it will be more convenient to consider $S$ as a single object. We also point out that the above conditions on $S$ can also be formulated in a manifestly basis-invariant manner \cite{Bischoff:2012}.

In the scalar case, conditions \refitem{item:YangBaxterEqn},  \refitem{item:TranslationInvariance}, \refitem{item:TCPInvariance} and \refitem{item:GaugeInvariance} drop out. In particular the absence of the Yang-Baxter equation \refitem{item:YangBaxterEqn} simplifies the structure significantly in that case, so that the general form of $S\in\SF$ can be worked out explicitly \cite{Lechner:2006}. For $\dim\K>1$, the general solution of the constraints summarized in Definition \ref{definition:SMatrix} is not known. However, many special solutions, corresponding to model theories such as $O(N)$ $\sigma$-models, are known and will be discussed in Section \ref{section:examples}.
\\
\\
In the following, it will not be necessary to rely on the detailed structure of particular solutions to the constraints summarized in Definition \ref{definition:SMatrix}. We therefore consider some arbitrary $S\in\SF$, and proceed to the construction of an associated quantum field theory. As $S$ will be fixed in the following, we do not reflect the $S$-dependence of various objects introduced subsequently in our notation.

The first step is the construction of a convenient Hilbert space. A look at the scalar case shows that different equivalent choices are possible -- compare the $S$-symmetric Fock space used in \cite{Lechner:2008} with the usual Bose/Fermi-Fock spaces used in \cite{Lechner:2011, Alazzawi:2012}. We will use the $S$-symmetric version here. In a different context, this construction was carried out by Liguori and Mintchev for $S$-matrices satisfying only conditions \refitem{item:Unitarity}--\refitem{item:YangBaxterEqn} of Definition \ref{definition:SMatrix} \cite{LiguoriMintchev:1995-1}. The more particular structure of conditions \refitem{item:TranslationInvariance}--\refitem{item:Crossing} will enter at a later stage, in analogy to the scalar case discussed in \cite{Lechner:2003}. We recall here this construction.

Starting from the single particle space \eqref{eq:OneParticleSpace}, we consider the $n$-fold tensor products $\Hil_1^{\ot n}=L^2(\Rl^n,d^n\bte)\ot\K^{\ot n}$, and introduce the operators, $n\in\Nl$, $k\in\{1,...,n-1\}$, $\Psi_n\in\Hil_1^{\ot n}$,
\begin{align}\label{eq:Dnk}
     (D_{n,k}\Psi_n)(\bte)
     :=
     S(\te_{k+1}-\te_k)_{n,k}\Psi_n(\te_1,...,\te_{k+1},\te_k,...,\te_n)
     \,,
\end{align}
where $\bte:=(\te_1,...,\te_n)$. Thanks to properties \refitem{item:Unitarity}--\refitem{item:YangBaxterEqn} of Definition \ref{definition:SMatrix}, these operators generate a representation of the permutation group $\frS_n$ of $n$ letters. As usual, we denote the transposition exchanging $k$ and $k+1$ by $\tau_k\in\frS_n$, and define for arbitrary $i_1,...,i_r\in\{1,...,n-1\}$
\begin{align}\label{eq:Dn}
	D_n(\tau_{i_1}\cdots\tau_{i_r}):=D_{n,i_1}\cdots D_{n,i_r}
	\,.
\end{align}

\begin{lemma}{\bf \cite{LiguoriMintchev:1995-1}}\label{lemma:Dn}
   $D_n$ \eqref{eq:Dn} is a unitary representation of $\frS_n$ on $\Hil_1^{\ot n}$.
\end{lemma}

At this point, we do not explicitly compute the representing operators $D_n(\pi)$. For a closed formula in the scalar case $\K=\Cl$, see \cite{Lechner:2006}. We only note here that by construction of $D_n$, for every $\pi\in\frS_n$ there exists a unitary tensor $\Rl^n\ni\bte\mapsto S_n^\pi(\bte)\in\U(\K^{\ot n})$ such that
\begin{align}\label{eq:SPi}
   (D_n(\pi)\Psi_n)(\bte)
   =
   S_n^\pi(\bte)\Psi_n(\te_{\pi(1)},...,\te_{\pi(n)})
   \,,\qquad
   \Psi_n\in\Hil_1^{\ot n}
   \,.
\end{align}
Clearly $S_n^{\tau_k}(\bte)=S(\te_{k+1}-\te_k)_{n,k}$, and for some other relevant permutations $\pi$, the tensor $S_n^\pi(\bte)$ will be calculated later on.

The orthogonal projection $P_n\in\B(\Hil_1^{\ot n})$ onto the $D_n$-invariant subspace will be denoted $P_n=\frac{1}{n!}\sum_{\pi\in\frS_n}D_n(\pi)$. We now define the {\it $S$-symmetrized Fock space} $\Hil$ over $\Hil_1$ as
\begin{align}\label{eq:Hil}
     \Hil:=\bigoplus_{n=0}^\infty \Hil_n
     \,,\qquad
     \Hil_n:=P_n\Hil_1^{\ot n},\;n\geq1\,,\;\;\Hil_0:=\Cl\,.
\end{align}
Its elements are thus sequences $\Psi=(\Psi_0,\Psi_1,...\,)$, $\Psi_n\in L^2(\Rl^n,d^n\bte)\ot\K^{\ot n}$, subject to the symmetry condition (here and in the following, we will make use of the summation convention)
\begin{align}\label{eq:SSymmetry}
     \Psi_n^\balpha(\bte)
     =
     S^{\alpha_k\alpha_{k+1}}_{\beta_k\beta_{k+1}}(\te_{k+1}-\te_k)
     \Psi_n^{\alpha_1...\alpha_{k-1}\beta_{k}\beta_{k+1}\alpha_{k+2}...\alpha_n}(\te_1,...,\te_{k+1},\te_k,...,\te_n)
     \,,
\end{align}
and having finite norm
\begin{align*}
     \|\Psi\|^2
     =
     \langle\Psi,\Psi\rangle
     =
     \sum_{n=0}^\infty \int d^n\bte\,(\Psi_n(\bte),\Psi_n(\bte))
     =
     \sum_{n=0}^\infty \int d^n\bte\,\sum_{\balpha}\overline{\Psi_n^\balpha(\bte)}\Psi_n^\balpha(\bte)
     <
     \infty\,.
\end{align*}
Occasionally we will also use the orthogonal projection $P:\bigoplus_n\Hil_1^{\ot n}\to\Hil$, the ``particle number operator'' $(N\Psi)_n:=n\Psi_n$, and the dense subspace $\DD\subset\Hil$ of ``finite particle number''. For the time being, these are just names for certain subspaces, their physical interpretation in terms of particle states will be justified later in scattering theory.

Next we discuss Poincar\'e symmetries and inner symmetries on $\Hil$. On the unsymmetrized Fock space $\widehat{\Hil}:=\bigoplus_n\Hil_1^{\ot n}$, we have a natural representation  $\widehat{U}$ of $\PGpo$ by second quantization of \eqref{eq:U1}, i.e.
\begin{align}\label{eq:Uhat}
   [\widehat{U}(a,\la)\Psi]_n^{\balpha}(\bte)
   :=
   \exp (i \sum^{n}_{l=1} \; p_{\alpha_l}(\te_l) \cdot a) \;\Psi_n^\balpha(\te_1-\la,...,\te_n-\la)
   \,,
\end{align}
where $p_{\alpha_l}$ is shorthand for $p_{m_{[\alpha_l]}}$. It is straightforward to verify that $\widehat{U}$ is a unitary strongly continuous positive energy representation of $\PGpo$, with up to a phase unique invariant unit vector $\Om$, given by $\Om_n:=\delta_{n,0}$. The representation $\widehat{U}$ can be extended to the full Poincar\'e group in various ways. However, only specific operators implementing space-, time-, and
space-time-reflection restrict to the $S$-symmetric subspace. We will only work with the proper Poincar\'e group here, and choose as an implementation of the space-time reflection $j(x):=-x$ the
anti-unitary TCP operator $\widehat{J}:\widehat{\Hil}\to\widehat{\Hil}$,
\begin{align}\label{eq:Jhat}
   (\widehat{J}\Psi)_n^\balpha(\bte)
   :=
   \overline{\Psi_n^{\overline{\alpha_n}...\overline{\alpha_1}}(\te_n,...,\te_1)}
   \,.
\end{align}
Clearly $\widehat{J}$ restricts to the single particle TCP operator \eqref{eq:J1} on $\Hil_1$ and satisfies $\widehat{J}\,^2=1$. Making use of $m_{[\overline{\alpha}]}=m_{[\alpha]}$, it is straightforward to check that $\widehat{U}(j):=\widehat{J}$ extends $\widehat{U}$ to a representation of the proper Poincar\'e group $\PGp$.

Also the inner symmetry group $G$ acts naturally on $\widehat{\Hil}$, via the unitary representation
\begin{align}\label{eq:VSecondQuantized}
	\widehat{V}
	:=
	\bigoplus_{n=0}^\infty V_1^{\ot n}
	\,.
\end{align}

\begin{lemma}\label{lemma:U}
     The representations $\widehat{U}$ of $\PGp$ and $\widehat{V}$ of $G$ commute and leave the subspace $\Hil\subset\widehat{\Hil}$ invariant. Their restrictions will be denoted
     \begin{align}\label{eq:U}
        U(a,\la):=\widehat{U}(a,\la)|_{\Hil}
        \,,\qquad
        J:=\widehat{J}|_{\Hil}
        \,,
        \qquad
        V(g):=\widehat{V}(g)|_{\Hil}
        \,.
     \end{align}
\end{lemma}
\begin{proof}
     As $\widehat{U}$ and $\widehat{V}$ preserve the grading of $\widehat{\Hil}$ w.r.t. the particle number, the claim about  restrictability follows once we established that $\widehat{U}(a,\la)$, $\widehat{J}$ and $\widehat{V}(g)$ commute with the projections $P_n$ onto the subspaces $\Hil_n\subset\Hil_1^{\ot n}$ invariant under the representation $D_n$ \eqref{eq:Dn} of the permutation group. For the boosts $(0,\la)$, this is clear since the $S$-matrix only depends on differences of rapidities in \eqref{eq:Dnk} and the transpositions generate $\frS_n$. For the translations $(a,0)$, we compute, $\Psi_n\in\Hil_1^{\ot n}$,
     \begin{align*}
	  ([U(a,0),D_{n,k}]\Psi)_n^\balpha(\bte)
	  &=
	  e^{i\sum^{n}_{l=1} p_{\alpha_l}(\te_l) a}
	  (1- e^{i(p_{\beta_k}(\te_{k+1})-p_{\alpha_{k+1}}(\te_{k+1}))a}e^{i(p_{\beta_{k+1}}(\te_{k})-p_{\alpha_{k}}(\te_{k}))a})
	  \\
	  &\qquad\times S^{\alpha_{k}\alpha_{k+1}}_{\beta_{k}\beta_{k+1}}(\te_{k+1}-\te_k)
	  \Psi_n^{\alpha_1...\beta_{k}\beta_{k+1}...\alpha_n}(\te_1,...,\te_{k+1},\te_k,...,\te_n)
	  \,.
     \end{align*}
     Thanks to Definition \ref{definition:SMatrix} \refitem{item:TranslationInvariance}, the $S$-matrix element vanishes unless the masses $m_{[\beta_k]}=m_{[\alpha_{k+1}]}$, $m_{[\beta_{k+1}]}=m_{[\alpha_{k}]}$ coincide, which implies identical on-shell momenta $p_{\beta_k}(\te_{k+1})=p_{\alpha_{k+1}}(\te_{k+1})$, $p_{\beta_{k+1}}(\te_{k})=p_{\alpha_{k}}(\te_{k})$ \eqref{eq:palpha}. Thus $[U(a,0),D_{n,k}]\Psi_n=0$, and since $k$ and $\Psi_n$ were arbitrary, we conclude $[U(a,0),P_n]=0$.

     The same argument shows also that $\widehat{V}(g)$ commutes with $\widehat{U}(a,\la)$. As $J_1$ was built in such a way that it commutes with $V_1(g)$ for any $g\in G$, it is clear that also $J$ and $V(g)$ commute. Furthermore, each $\widehat{V}(g)$ restricts to $\Hil$ since $V_1(g)\ot V_1(g)$ commutes with $S(\te)$.

     Concerning $\widehat{J}$, we first note that written in components, the $S$-matrix properties \refitem{item:Unitarity}, \refitem{item:HermitianAnalyticity} and \refitem{item:TCPInvariance} of Definition \ref{definition:SMatrix} combine to the identity
     \begin{align*}
	  \overline{S^{\alpha\beta}_{\gamma\delta}(\te)}
	  =
	  S^{\overline{\beta}\overline{\alpha}}_{\overline{\delta}\overline{\gamma}}(-\te)
	  \,.
     \end{align*}
     With this information we compute
     \begin{align*}
	  [\widehat{J}\,D_{n,k}\Psi_n]^\balpha(\bte)
	  &=
	  \overline{[D_{n,k}\Psi_n]^{\overline{\alpha_n}...\overline{\alpha_1}}(\te_n,...,\te_1)}
	  \\
	  &=
	  \overline{S^{\overline{\alpha_{n-k+1}}\,\overline{\alpha_{n-k}}}_{\overline{\beta_{n-k+1}}\,\overline{\beta_{n-k}}}(\te_{n-k}-\te_{n-k+1})\,\Psi_n^{\overline{\alpha_n}...\overline{\beta_{n-k+1}}\,\overline{\beta_{n-k}}...\overline{\alpha_1}} (\te_n,..,\te_{n-k},\te_{n-k+1},..,\te_1)}
	  \\
	  &=
	  S_{\beta_{n-k}\,\beta_{n-k+1}}^{\alpha_{n-k}\,\alpha_{n-k+1}}(\te_{n-k+1}-\te_{n-k})
	  \overline{\Psi_n^{\overline{\alpha}_n...\overline{\beta}_{n-k+1}\,\overline{\beta}_{n-k}...\overline{\alpha}_1} (\te_n,..,\te_{n-k},\te_{n-k+1},..,\te_1)}
	  \\
	  &=
	  [D_{n,n-k}\widehat{J}\,\Psi_n]^\balpha(\bte)\,.
     \end{align*}
     This commutation relation implies that $\widehat{J}$ commutes with the average $P_n$ over $\frS_n$.
\end{proof}

For our purposes, the TCP operator $J$ is an important object as it will allow us to connect localization regions extending to left and right spacelike infinity, and incoming with outgoing scattering states. Under further assumptions on $S$ (see, for example, \cite[p.~12]{Dorey:1998}), one can also build models in which the S-matrix is invariant under the symmetries of time reflection and parity separately. In this case, also an extension of $\widehat{U}$ to the full Poincar\'e group can be restricted to $\Hil$. But these more particular properties will not be relevant here.
\\
\\
As a prerequisite for the definition of quantum fields, we also recall the structure of creation and annihilation operators on the $S$-symmetric Fock space $\Hil$. On $\widehat{\Hil}$, we have the usual unsymmetrized operators $a(\varphi)$, $\ad(\varphi)$, $\varphi\in\Hil_1$. They are defined by linear and continuous extension from
\begin{align}
   \ad(\varphi)\psi_1\ot...\ot\psi_n
   &:=
   \sqrt{n+1}\,\varphi\ot\psi_1\ot...\ot\psi_n\,,\qquad\psi_1,...,\psi_n\in\Hil_1\,,
   \label{eq:AD}
   \\
   a(\varphi)\psi_1\ot...\ot\psi_n
   &:=
   \sqrt{n}\,\langle\varphi,\psi_1\rangle\,\psi_2\ot...\ot\psi_n
   \,,\qquad\;\;
   a(\varphi)\Om:=0\,,
   \label{eq:A}
\end{align}
to $\Hil_n$, and then to the subspace of finite particle number, where they satisfy $a(\varphi)^*\supset\ad(\varphi)$. We introduce their projections onto $\Hil$ as
\begin{align}\label{eq:zzd}
     \zd(\varphi):=P\ad(\varphi)P\,,\qquad z(\varphi):=Pa(\varphi)P\,,\qquad \varphi\in\Hil_1,
\end{align}
and the distributional kernels $z^\#_\alpha(\te)$ related to these operators by
\begin{align}\label{eq:ZZDKernels}
     \zd(\varphi)
     =
     \int d\te\,\varphi_\alpha(\te)\zd_\alpha(\te)
     \,,\qquad
     z(\varphi)
     =
     \int d\te\,\overline{\varphi_\alpha(\te)}z_\alpha(\te)
     \,.
\end{align}

\begin{proposition}\label{proposition:ZZD}
     Let $\varphi\in\Hil_1$ and $\Psi\in\DD$ be arbitrary.
     \begin{propositionlist}
	  \item\label{item:ZZDexplicit} The operators \eqref{eq:zzd} are explicitly given by
	  \begin{align}\label{eq:Zexplicit}
	       [z(\varphi) \Psi]_n^\balpha(\bte)
	       &=
	       \sqrt{n+1} \int d\te'\,\overline{\varphi^{\beta}(\te')}
	       \Psi_{n+1}^{\beta\balpha}(\te',\bte),
	       \\
	       \label{eq:ZDexplicit}
	       [\zd(\varphi)\Psi]_n(\bte)
	       &=
	       \frac{1}{\sqrt{n}}\sum_{k=1}^n S_n^{\sigma_k}(\bte)(\varphi(\te_k)\ot\Psi_{n-1}(\te_1,...,\hat{\te}_k,...,\te_n)),\qquad n\geq1\,,
	       \\
	       [\zd(\varphi)\Psi]_0
	       &=
	       0
	       \,,
	       \label{eq:ZD0}
	  \end{align}
	  where $\hat{\te}_k$ means that this variable is omitted, and the permutations $ \sigma_k\in\frS_n$ are defined as $\sigma_k:=\tau_{k-1}\tau_{k-2} \cdots \tau_1$ for $k\geq1$, and $\sigma_1:={\rm id}$.
	  \item\label{item:ZZDStar} We have
	  \begin{align}\label{eq:ZZDStar}
	       z(\varphi)^*\supset\zd(\varphi)\,.
	  \end{align}
	  \item\label{item:ZZDCovariance} For $(a,\la)\in\PGpo$ and $g\in G$, we have
	  \begin{align*}
	     U(a,\la)z^\#(\varphi)U(a,\la)^{-1}
	     &=
	     z^\#(U(a,\la)\varphi)
	     \,,\\
	     V(g)z^\#(\varphi)V(g)^{-1}
	     &=
	     z^\#(V(g)\varphi)
	     \,,
	  \end{align*}
	   where $z^\#$ stands for either $z$ or $\zd$.
	  \item\label{item:ZZDNBounds} With respect to the particle number operator $N$, there hold the bounds
	  \begin{align}\label{eq:NBoundsOnZZD}
	       \|z(\varphi)\Psi\| \leq \|\varphi\| \|N^{1/2}\Psi\|, \qquad \|z^{\dagger}(\varphi)\Psi\| \leq \|\varphi\| \|(N+1)^{1/2}\Psi\|\,.
	  \end{align}
	  \item\label{item:ZFAlgebra} The distributional kernels $z^\#_\alpha(\te)$ satisfy
	  \begin{align}
	       z_{\alpha}(\te) z_{\beta}(\te') - S^{\beta\alpha}_{\delta\gamma}(\te-\te') z_{\gamma}(\te') z_{\delta}(\te)  &=0,
	       \label{eq:ZZ-Kernel-Commutator}
	       \\
	       \zd_{\alpha}(\te) \zd_{\beta}(\te') - S^{\gamma\delta}_{\alpha\beta}(\te-\te') \zd_{\gamma}(\te')\zd_{\delta}(\te)  &=0,
	       \\
	       z_{\alpha}(\te) \zd_{\beta}(\te') - S^{\alpha\gamma}_{\beta\delta}(\te'-\te) \zd_{\gamma}(\te') z_{\delta}(\te)
	       &=
	       \delta^{\alpha\beta}\delta(\te-\te')\cdot 1\,.
	  \end{align}
     \end{propositionlist}
\end{proposition}
\begin{proof}
     \refitem{item:ZZDexplicit} By comparison of \eqref{eq:A} and \eqref{eq:Zexplicit}, one observes that the explicit action of $a(\varphi)$ on the finite particle number subspace of $\widehat{\Hil}$ is given by the formula \eqref{eq:Zexplicit}. But this action does not disturb the symmetrization \eqref{eq:SSymmetry}. Hence $a(\varphi)$ leaves $\DD\subset\widehat{\Hil}$ invariant, which implies $z(\varphi)\Psi=Pa(\varphi)P\Psi=a(\varphi)\Psi$, $\Psi\in\DD$, and the claimed formula \eqref{eq:Zexplicit} follows.

     To compute $\zd(\varphi)$, we first note that by definition of this operator, $[\zd(\varphi)\Psi]_n=\sqrt{n}P_n(\varphi\ot\Psi_{n-1})$ for $n\geq1$, and $[\zd(\varphi)\Psi]_0=0$. This implies in particular \eqref{eq:ZD0}. Each permutation $\pi\in\frS_n$, $n\geq2$, can be decomposed according to $\pi = \sigma_k \rho$ with some $\sigma_k=\tau_{k-1}\tau_{k-2} \cdots\tau_{1} \in \frS_n$ and some $\rho\in\frS_{n-1}$ acting on $\{2,...,n\}\subset\{1,...,n\}$.  On the level of the projection $P_n$, this gives $P_n=\frac{1}{n}\sum_{k=1}^nD_n(\sigma_k)(1\ot P_{n-1})$ \cite{Lechner:2003}. Taking into account \eqref{eq:SPi} yields the claimed formula \eqref{eq:ZDexplicit}.

     \refitem{item:ZZDStar}--\refitem{item:ZZDNBounds} are known to hold for the unsymmetrized creation and annihilation operators $a^\#(\varphi)$, with $U$ replaced by $\widehat{U}$ and $V$ replaced by $\widehat{V}$. Since $z^\#(\varphi)=Pa^\#(\varphi)P$ and $P$ is selfadjoint, \refitem{item:ZZDStar} follows. As $\widehat{U}(a,\la)$, $\widehat{V}(g)$ commute with $P$ and equals $U(a,\la)$, $V(g)$ on $\Hil$, we also have \refitem{item:ZZDCovariance}. The fact that the norm of the orthogonal projection $P$ is $\|P\|=1$ implies \refitem{item:ZZDNBounds}.

     \refitem{item:ZFAlgebra}:  These commutation of distributional kernels can be deduced from the formulas in \refitem{item:ZZDexplicit}, cf. also \cite{LiguoriMintchev:1995-1}. For example, one sees by comparing \eqref{eq:Zexplicit} and \eqref{eq:ZZDKernels} that
	\begin{align*}
		[z_\alpha(\te)\Psi]_n^\btau(\tilde{\bte})
		&=
		\sqrt{n+1}\Psi_{n+1}^{\alpha\btau}(\te,\tilde{\bte})
		\,.
	\end{align*}
	Taking into account the symmetry \eqref{eq:SSymmetry}, we thus find
	\begin{align*}
		[z_\alpha(\te)z_\beta(\te')\Psi]_n^\btau(\tilde{\bte})
		&=
		\sqrt{(n+1)(n+2)}\Psi_{n+2}^{\beta\alpha\btau}(\te',\te,\tilde{\bte})
		\\
		&=
		\sqrt{(n+1)(n+2)}\,S^{\beta\alpha}_{\delta\gamma}(\te-\te')\Psi_{n+2}^{\delta\gamma\btau}(\te,\te',\tilde{\bte})
		\\
		&=
		S^{\beta\alpha}_{\delta\gamma}(\te-\te')\,[z_\gamma(\te')z_\delta(\te)\Psi]_n^{\btau}(\tilde{\bte})
		\,,
	\end{align*}
	which implies \eqref{eq:ZZ-Kernel-Commutator}. The derivation of the other two exchange relations is analogous.
\end{proof}

The algebraic relations in item \refitem{item:ZFAlgebra} are known as the {\em Zamolodchikov--Faddeev} algebra \cite{ZamolodchikovZamolodchikov:1979, Faddeev:1984}, and are frequently used in the context of integrable quantum field theories (see for example \cite{BabujianFoersterKarowski:2006, Smirnov:1992}, and references cited therein). Note in particular that for the constant $S$-matrices $S^{\alpha\beta}_{\gamma\delta}(\te)=\pm\delta^{\alpha}_\delta\delta^\beta_\gamma$, they coincide with the familiar CCR/CAR relations.
\\
\\
The covariance statements in item \refitem{item:ZZDCovariance} do {\em not} extend to the TCP operator $J$. In the next section, we will also need the TCP-transformed creation and annihilation operators,
\begin{align}
     \zd(\varphi)':=J\zd(J\varphi)J
     \,,\qquad
     z(\varphi)':=Jz(J\varphi)J
     \,.
\end{align}
Taking into account that $J$ is an antiunitary involution with $JU(a,\la)J=U(-a,\la)$, it becomes apparent that items \refitem{item:ZZDStar}--\refitem{item:ZZDNBounds} of Proposition \ref{proposition:ZZD} apply to the $z^\#(\varphi)'$ without any changes. The explicit actions in \refitem{item:ZZDexplicit} and the exchange relations in \refitem{item:ZFAlgebra} are different for the $z^\#(\varphi)'$, however. For example, the ``TCP reflected'' annihilation operator acts according to
\begin{align}
     [z(\varphi)'\Psi]_n^\balpha(\bte)
     &=
     \overline{[z(J\varphi)J\Psi]_n^{\overline{\alpha_n}...\overline{\alpha_1}}(\te_n,...,\te_1)}
     \nonumber
     \\
     &=
     \sqrt{n+1} \int d\te'\, \overline{\varphi^{\overline{\beta}}(\te')}
     \overline{(J\Psi)_{n+1}^{\beta\overline{\alpha_n}...\overline{\alpha_1}}(\te',\te_n,...,\te_1)},
     \nonumber
     \\
     &=
     \sqrt{n+1} \int d\te'\, \overline{\varphi^{\beta}(\te')}
     \Psi_{n+1}^{\balpha\beta}(\bte,\te'),
     \label{eq:Z'explicit}
\end{align}
it ``annihilates from the right''. Note that in the last step of the above calculation, we have used that the charge conjugation $\beta\mapsto\overline{\beta}$ is an involution.

The TCP-reflected creation/annihilation operators satisfy commutation relations analogous to the ones listed in Proposition \ref{proposition:ZZD} \refitem{item:ZFAlgebra} with the only difference that $S^{\alpha\beta}_{\gamma\delta}(\te)$ has to be replaced by $\overline{S^{\overline{\alpha}\overline{\beta}}_{\overline{\gamma}\overline{\delta}}(\te)}=S^{\beta\alpha}_{\delta\gamma}(-\te)$. As it will turn out, even more important than the exchange relations of the $z^\#(\varphi)$ and $z^\#(\varphi)'$ amongst each other are their {\em relative} commutation relations. They are determined next.

\begin{proposition}\label{proposition:ZZ'CommutationRelations}
     Let $\varphi_1,\varphi_2\in\Hil_1$ and $\Psi_n\in\Hil_n$, $n\in\Nl_0$. Then
     \begin{align}
	  [z(\varphi_1)',\, z(\varphi_2)]\Psi_n
	  &=
	  0,\label{eq:ZZ'Commutator}
	  \\
	  [\zd(\varphi_1)',\, \zd(\varphi_2)]\Psi_n
	  &=0,
	  \label{eq:ZDZD'Commutator}
	  \\
	  [z(\varphi_1)',\, \zd(\varphi_2)]\Psi_n
	  &=
	  K_n^{\varphi_1\,\varphi_2}\Psi_n,
	  \label{eq:Z'ZDCommutator}
	  \\
	  [\zd(\varphi_1)',\, z(\varphi_2)]\Psi_n
	  &=
	  L_n^{\varphi_1\,\varphi_2}\Psi_n,
	  \label{eq:ZD'ZCommutator}
     \end{align}
     where $K_n^{\varphi_1\,\varphi_2}$ and $L_n^{\varphi_1\,\varphi_2}$ are operators on $\Hil_n$ which act by multiplication with the tensors
     \begin{align}\label{eq:MixedCommutator1}
	  K_n^{\varphi_1\varphi_2}(\bte)^\balpha_\bbeta
	  &=
	  +\int d\te'\,\overline{\varphi_1^{{\gamma}}(\te')}
	  S_{n+1}^{\sigma_{n+1}}(\bte,\te')^{\balpha\gamma}_{\delta\bbeta}\varphi_2^{\delta}(\te')
	  \,,\\
	  L_n^{\varphi_1 \varphi_2}(\bte)^\balpha_\bbeta
	  &=
	  -\int d\te'\,\varphi_1^{{\gamma}}(\te')
	  \overline{S_{n+1}^{\sigma_{n+1}}(\bte,\te')^{\bbeta\gamma}_{\delta\balpha}}\overline{\varphi_2^{\delta}(\te')}
	  \,.
	  \label{eq:MixedCommutator2}
	\end{align}
\end{proposition}
\begin{proof}
   The first commutation relation \eqref{eq:ZZ'Commutator} can be computed straightforwardly on the basis of \eqref{eq:Zexplicit} and \eqref{eq:Z'explicit}: Since these annihilation operators contract the arguments of $\Psi_n$ from the left and right, respectively, they commute. As the creation operators in \eqref{eq:ZDZD'Commutator} are the adjoints of the annihilation operators in \eqref{eq:ZZ'Commutator}, the commutation relation \eqref{eq:ZDZD'Commutator} follows by taking the adjoint of \eqref{eq:ZZ'Commutator}.

   The verification of the mixed commutation relations requires a calculation. With $\Psi\in\DD$, inserting \eqref{eq:Z'explicit} and \eqref{eq:ZDexplicit} yields
   \begin{align*}
   	&([z(\varphi_1)',\zd(\varphi_2)]\Psi)_n^\balpha(\bte)
   	\\
   	&=
   	\sqrt{n+1}\int d\te_{n+1}\,\overline{\varphi_1^{{\gamma}}(\te_{n+1})}[\zd(\varphi_2)\Psi]_{n+1}^{\balpha\gamma}(\bte,\te_{n+1})
   	\\
   	&\qquad-
   	\frac{1}{\sqrt{n}}\sum_{k=1}^n
   	S_n^{\sigma_k}(\bte)^\balpha_{\delta\beta_1...\beta_{n-1}}
   	\varphi_2^\delta(\te_k)
   	[z(\varphi_1)'\Psi]_{n-1}^{\beta_1...\beta_{n-1}}(\te_1,...\hat{\te}_k,...,\te_n)
   	\\
   	&=
   	\int d\te_{n+1}\,\overline{\varphi_1^{{\gamma}}(\te_{n+1})}
	\sum_{k=1}^{n+1}
	S_{n+1}^{\sigma_k}(\bte,\te_{n+1})^{\balpha\gamma}_{\delta\beta_1...\beta_n}\varphi_2^{\delta}(\te_k)\Psi_n^\bbeta(\te_1,..,\hat{\te}_k,..,\te_{n+1})
	\\
	&\quad-
	\sum_{k=1}^n S_n^{\sigma_k}(\bte)^{\balpha}_{\delta\beta_1...\beta_{n-1}}
	\varphi_2^{\delta}(\te_k)
	\int d\te_{n+1}\,\overline{\varphi_1^{{\beta_n}}(\te_{n+1})}
	\Psi_n^{\beta_1...\beta_n}(\te_1,...,\hat{\te}_k,..,\te_n,\te_{n+1})
	\,.
   \end{align*}
   For terms corresponding to $k<n+1$ in the first sum, the permutation $\sigma_k$ satisfies $\sigma_k(n+1)=n+1$ and can be regarded as an element of $\frS_n$, with $S_{n+1}^{\sigma_k}(\bte,\te_{n+1})^{\balpha\gamma}_{\delta\beta_1...\beta_n}=\delta^\gamma_{\beta_n}S_n^{\sigma_k}(\bte)^{\balpha}_{\delta\beta_1...\beta_{n-1}}$. It is thus apparent that all terms with $k<n+1$ drop out in the above expression, and we are left with the term corresponding to $k=n+1$, which equals \eqref{eq:MixedCommutator1}.

     To compute the last commutator \eqref{eq:MixedCommutator2}, note that in view of $z(\varphi)^*\supset\zd(\varphi)$ and $(z(\varphi)')^*\supset\zd(\varphi)'$ we have
     \begin{align*}
	  L^{\varphi_1\varphi_2}_n(\bte)^\balpha_\bbeta\Psi_n^\bbeta(\bte)
	  &=
	  ([\zd(\varphi_1)',\,z(\varphi_2)]\Psi)_n^\balpha(\bte)
	  =
	  -([z(\varphi_1)',\,\zd(\varphi_2)]^*\Psi)_n^\balpha(\bte)
	  \\
	  &=
	  -(K_n^{\varphi_1\varphi_2}(\bte)^*)^\balpha_\bbeta\Psi_n^\bbeta(\bte)
	  =
	  -\overline{K_n^{\varphi_1\varphi_2}(\bte)_\balpha^\bbeta}\Psi_n^\bbeta(\bte)
	  \,.
     \end{align*}
     Thus $L^{\varphi_1\varphi_2}_n(\bte)^\balpha_\bbeta= -\overline{K_n^{\varphi_1\varphi_2}(\bte)_\balpha^\bbeta}$, from which we read off \eqref{eq:MixedCommutator2}.
\end{proof}

Later on, we will also need the explicit form of the components of the tensors $S_{n+1}^{\sigma_{n+1}}(\bte,\te')$.
\begin{lemma}
	Let $n\in\Nl$ and $k\in\{1,...,n\}$. Then
	\begin{align}\label{eq:SSigma_k}
		S_n^{\sigma_k}(\bte)^\balpha_\bbeta
		&=
		\sum_{\xi_1,...,\xi_k}
		\delta^{\alpha_k}_{\xi_k}\delta^{\beta_1}_{\xi_1}
		\prod_{l=1}^{k-1}
		S^{\alpha_l\xi_{l+1}}_{\xi_l\beta_{l+1}}(\te_k-\te_l)
		\cdot
		\delta^{\alpha_{k+1}}_{\beta_{k+1}}\cdots\delta^{\alpha_{n}}_{\beta_{n}}
		\,.
\end{align}
\end{lemma}
\begin{proof}
	We will proceed by induction in $k$. For $k=1$, we have $\sigma_1={\rm id}$, implying $D_n(\sigma_1)=1$ and thus $S_n^{\sigma_1}(\bte)=1_n$. This is reproduced by \eqref{eq:SSigma_k}:
	\begin{align*}
		S_n^{\sigma_1}(\bte)^\balpha_\bbeta
		&=
		\sum_{\xi_1}
		\delta^{\alpha_1}_{\xi_1}\delta^{\beta_1}_{\xi_1}
		\cdot
		\delta^{\alpha_{2}}_{\beta_{2}}\cdots\delta^{\alpha_{n}}_{\beta_{n}}
		=
		\delta^{\alpha_1}_{\beta_1}\cdots\delta^{\alpha_{n}}_{\beta_{n}}
		=
		(1_n)^\balpha_\bbeta
		\,.
	\end{align*}
	For the induction step $k\to k+1$, we take into account $\sigma_{k+1}=\tau_k\sigma_k$ and calculate, $\Psi_n\in\Hil_n$, writing $\te_{a,b}:=\te_a-\te_b$ as a shorthand notation,
	\begin{align*}
		[&D_n(\sigma_{k+1})\Psi_n]^\balpha(\bte)
		=
		S^{\alpha_k\alpha_{k+1}}_{\eps\beta_{k+1}}(\te_{k+1,k})
		[D_n(\sigma_k)\Psi_n]^{\alpha_1...\alpha_{k-1}\eps\beta_{k+1}...\alpha_n}(\te_1,...,\te_{k+1},\te_k,...,\te_n)
		\\
		&=
		S^{\alpha_k\alpha_{k+1}}_{\eps\beta_{k+1}}(\te_{k+1,k})
		\sum_{\xi_1,...,\xi_k}
		\delta^{\eps}_{\xi_k}\delta^{\beta_1}_{\xi_1}
		\prod_{l=1}^{k-1}
		S^{\alpha_l\xi_{l+1}}_{\xi_l\beta_{l+1}}(\te_{k+1,l})
		\Psi_n^{\beta_1...\beta_{k+1}\alpha_{k+2}..\alpha_n}(\te_{k+1},\te_1,..,\hat{\te}_{k+1},..,\te_n)
		\\
		&=
		\sum_{\xi_1,...,\xi_k}
		\delta^{\alpha_{k+1}}_{\xi_{k+1}}
		\delta^{\beta_1}_{\xi_1}
		S^{\alpha_k\xi_{k+1}}_{\xi_k\beta_{k+1}}(\te_{k+1,k})
		\prod_{l=1}^{k-1}
		S^{\alpha_l\xi_{l+1}}_{\xi_l\beta_{l+1}}(\te_{k+1,l})
		\Psi_n^{\beta_1...\beta_{k+1}\alpha_{k+2}..\alpha_n}(\te_{k+1},\te_1,..,\hat{\te}_{k+1},..,\te_n)
		.
	\end{align*}
	The last line coincides with $S_n^{\sigma_{k+1}}(\bte)^\balpha_\bbeta\Psi_n^\bbeta(\te_{k+1},\te_1,..,\hat{\te}_{k+1},..,\te_n)$ when $S^{\sigma_{k+1}}_n(\bte)^\balpha_\bbeta$ is given by \eqref{eq:SSigma_k}. As the tensors $S_n^\pi(\bte)$ were defined by this equation \eqref{eq:SPi}, the proof is finished.
\end{proof}

For later reference, we note here that with the above result, we have found in particular an explicit form of the integral kernel appearing in \eqref{eq:MixedCommutator1} and \eqref{eq:MixedCommutator2}, namely
\begin{align}\label{eq:SSigma_nExplicit}
	S_{n+1}^{\sigma_{n+1}}(\bte,\te')^{\balpha\gamma}_{\delta\bbeta}
	&=
	\sum_{\xi_1,...,\xi_{n+1}}
	\delta^{\gamma}_{\xi_{n+1}}\delta^{\delta}_{\xi_1}
	\prod_{l=1}^n
	S^{\alpha_l\xi_{l+1}}_{\xi_l\beta_l}(\te'-\te_l)
	\,.
\end{align}

\section{Multi-component wedge-local fields}\label{section:fields}

We now turn to the definition of two quantum fields $\phi,\phi'$ on two-dimensional Minkowski space which are auxiliary but important objects in our construction. As before, we will assume a fixed $S$-matrix $S\in\SF$ and suppress all dependence on $S$ in our notation.

Similar to the single particle space $\Hil_1=L^2(\Rl,d\te)\ot\K$, we will deal with test functions $f\in\Ss(\Rl^2)\ot\K$ having several components $x\mapsto f_\alpha(x):=(e^\alpha,f(x))$. Their rapidity space wave functions are defined as
\begin{align}\label{eq:fpm}
	f^\pm_\alpha(\te)
	:=
	\fti_\alpha(\pm p_\alpha(\te))
	=
	\frac{1}{2\pi}
	\int d^2x\,f_\alpha(x)\,e^{\pm ip_\alpha(\te)\cdot x}
	\,.
\end{align}
As before, $p_\alpha(\te)=m_{[\alpha]}(\cosh\te,\,\sinh\te)$ denotes the momentum on the mass shell with mass $m_{[\alpha]}$. Clearly, $f^\pm_\alpha\in L^2(\Rl,d\te)$ for $f_\alpha\in\Ss(\Rl^2)$, and so we may consider $f^\pm$ as vectors in $\Hil_1$. Obviously the maps $\Ss(\Rl^2)\ot\K\ni f\mapsto f^\pm\in\Hil_1$ are linear, and it is not difficult to see that they are continuous as well.

On $\Ss(\Rl^2)\ot\K$, we have a natural action of the proper orthochronous Poincar\'e group by pullback, $f\mapsto (a,\la)\rhd f:= f\circ(a,\la)^{-1}$, $(a,\la)\in\PGpo$. Here $(a,\la)x:=\La_\la x+a$, where $\La_\la$ denotes the Lorentz boost matrix with rapidity $\la$. This action extends to the proper Poincar\'e group by implementing spacetime reflection $j(x):=-x$ as the TCP type map
\begin{align*}
	(j\rhd f)_\alpha(x)
	:=
	\overline{f_{\overline{\alpha}}(-x)}
	\,.
\end{align*}
Taking into account $m_{[\overline{\alpha}]}=m_{[\alpha]}$, it is straightforward to check with these definitions that $f\mapsto f^+$ and $f\mapsto Jf^-$ are covariant in the sense that
\begin{align}\label{eq:Equivariance}
	(g\rhd f)^+
	=
	U(g)f^+
	\,,\qquad
	J(g\rhd f)^-
	=
	U(g)Jf^-
	\,,\qquad
	g\in\PGp.
\end{align}

The fields we are interested in are defined as, $f\in\Ss(\Rl^2)\ot\K$,
\begin{align}
	\phi(f)
	&:=
	\zd(f^+)+z(Jf^-)
	\,,\\
	\phi'(f)
	&:=
	\zd(f^+)'+z(Jf^-)'
	\,.
\end{align}

\begin{proposition}\label{proposition:PhiPhi'}
	Let $f\in\Ss(\Rl^2)\ot\K$ and $\Psi\in\DD$.
	\begin{propositionlist}
		\item\label{item:PhiLinearAndContinuous} $\Ss(\Rl^2)\ot\K\ni f\mapsto \phi(f)\Psi$ is linear and continuous.
		\item\label{item:PhiAdjoint} $\phi(f)^*\supset\phi(f^*)$ with $(f^*)_\alpha(x):=\overline{f_{\overline{\alpha}}(x)}$.
		\item\label{item:PhiAnalyticVectors} All vectors in $\DD$ are entire analytic for $\phi(f)$. For $f=f^*\in\Ss(\Rl^2)\ot\K$, the field operator $\phi(f)$ is essentially selfadjoint.
		\item\label{item:PhiCovariance} $\phi$ transforms covariantly under $\PGpo$, i.e.,
		\begin{align}\label{eq:PhiCovariance}
			\phi((a,\la)\rhd f)\Psi
			=
			U(a,\la)\phi(f)U(a,\la)^{-1}\Psi
			\,,\qquad
			(a,\la)\in\PGpo.
		\end{align}
		\item\label{item:PhiPhi'Covariance} $\phi(j\rhd f)=J\phi'(f)J$.
		\item\label{item:PhiReehSchlieder} The vacuum vector $\Om$ is cyclic for the field $\phi$, i.e. for any open set $\OO\subset\mathbb{R}^2$, the subspace
		\begin{equation*}
			\DD_{\OO}
			:=
			{\rm span}\{\phi(f_1)\cdots\phi(f_n)\Omega\;:\; f_1,...,f_n\in\Ss(\OO)\otimes\K,\,n\in\Nl_0\}
		\end{equation*}
		is dense in $\Hil$.
		\item\label{item:PhiGaugeSymmetries}
		For $g\in G$, $f\in\Ss(\Rl^2)\ot\K$, let $(V_1(g)f)(x):=V_1(g)f(x)$. Then
		\begin{align}
			V(g)\phi(f)V(g)^{-1}
			&=
			\phi(V_1(g)f)
			\,.
		\end{align}
		\item\label{item:PhiNotLocal} $\phi$ is local if and only if $S^{\alpha\beta}_{\gamma\delta}(\te)=\delta^\alpha_\delta\delta^\beta_\gamma$.
	\end{propositionlist}
	All these statements also hold if $\phi$ is replaced by $\phi'$ (and $\phi'$ by $\phi$ in item \refitem{item:PhiPhi'Covariance}).
\end{proposition}
\begin{proof}
	\refitem{item:PhiLinearAndContinuous} As $\zd$ is linear, whereas $z$ and $J$ are antilinear, the linearity of $\phi$ is clear. The continuity of $f\mapsto\phi(f)\Psi$ follows from the continuity of $\Ss(\Rl^2)\ot\K\ni f\mapsto Jf^\pm\in\Hil_1$ and the continuity of $\Hil_1\ni\xi\mapsto z^{\#}(\xi)\Psi\in\Hil$ (Prop.~\ref{proposition:ZZD} \refitem{item:ZZDNBounds}).

	\refitem{item:PhiAdjoint} By definition of $f^*$, there holds $(f^*)^\pm=Jf^\mp$. In view of Prop.~\ref{proposition:ZZD} \refitem{item:ZZDStar}, we have
	\begin{align*}
		\phi(f)^*\Psi
		=
		z(f^+)\Psi+\zd(Jf^-)\Psi
		=
		z(J(f^*)^-)\Psi+\zd((f^*)^+)\Psi
		=
		\phi(f^*)\Psi\,.
	\end{align*}
	\refitem{item:PhiAnalyticVectors} This argument is identical to the scalar case, see \cite[Prop.~1(2)]{Lechner:2003}. \refitem{item:PhiCovariance} follows at once from the above mentioned covariance of $f\mapsto f^+$, $f\mapsto Jf^-$, and the covariance of the $z^{\#}$ expressed in Prop.~\ref{proposition:ZZD}~\refitem{item:ZZDCovariance}. For \refitem{item:PhiPhi'Covariance}, we compute
	\begin{align*}
		\phi(j\rhd f)\Psi
		=
		(\zd(Jf^+)+z(JJf^-))\Psi
		=
		(J\zd(f^+)'J+Jz(Jf^-)'J)\Psi
		=
		J\phi'(f)J\Psi\,.
	\end{align*}
	\refitem{item:PhiReehSchlieder} By standard analyticity arguments making use of the fact that $U$ is a positive energy representation of $\PG_+$, it follows that $\DD_\OO\subset\Hil$ is dense if and only if $\DD_{\Rl^2}\subset\Hil$ is dense. Now for any scalar test function $f_0\in\Ss(\Rl^2)$ with $\supp\fti_0$ contained in the forward light cone, and any $\alpha_0\in\I$, the function $f:=f_0\ot e^{\alpha_0}\in\Ss(\Rl^2)\ot\K$ satisfies $f^+_\alpha=\delta_{\alpha\alpha_0}\cdot f_{\alpha_0}^+$ and $f^-_\alpha=0$, $\alpha\in\I$. Thus $\phi(f)=\zd(f^+)$. Furthermore, as $f_0$ and $\alpha_0$ vary within the above limitations, the space spanned by all corresponding $f^+$ is dense in $\Hil_1$. Since the creation operators $\zd$ generate the $S$-symmetric Fock space, this implies that $\DD_{\Rl^2}\subset\Hil$ is dense.

	\refitem{item:PhiGaugeSymmetries} This follows from Proposition~\ref{proposition:ZZD}~\refitem{item:ZZDCovariance} and the fact that $V(g)$ and $J$ commute.

 	\refitem{item:PhiNotLocal} In case the $S$-matrix is $S^{\alpha\beta}_{\gamma\delta}(\te)=\delta^\alpha_\delta\delta^\beta_\gamma$, the $S$-symmetric Fock space is just the completely symmetric Bose Fock space over $\Hil_1$, as can be read off from \eqref{eq:SSymmetry}. It is clear from our construction that in this case, $\phi$ is the free field, which is known to be local. To show that this form of $S$ is also necessary for locality, we consider the two-particle contribution of the field commutator on the vacuum. With $f\in\Ss(\OO_1)\ot\K$, $g\in\Ss(\OO_2)\ot\K$ with spacelike separated open regions $\OO_1,\OO_2\subset\Rl^2$, we have $[[\phi(f),\phi(g)]\Om]_2=[\zd(f^+),\zd(g^+)]\Om=\sqrt{2}P_2(f^+\ot g^+-g^+\ot f^+)$. According to part \refitem{item:PhiReehSchlieder}, $f^+$ and $g^+$ span dense subspaces of $\Hil_1$ as $f,g$ vary within $\Ss(\OO_1)\ot\K$ and $\Ss(\OO_2)\ot\K$, respectively. Hence for any $f,g\in\Ss(\Rl^2)\ot\K$, and any $\alpha,\beta\in\I$, we must have
	\begin{align*}
		0
		&=
		\sqrt{2}[[\phi(f),\,\phi(g)]\Omega]_2^{\alpha\beta}(\theta_1,\theta_2)
		\\
		&=
		f^+_\alpha(\te_1) g^+_\beta(\te_2) + S^{\alpha\beta}_{\gamma\delta}(\te_2-\te_1)
		f^+_\gamma(\te_2) g^+_\delta(\te_1)
		-
		g^+_\alpha(\te_1) f^+_\beta(\te_2) - S^{\alpha\beta}_{\gamma\delta}(\te_2-\te_1)	g^+_\gamma(\te_2) f^+_\delta(\te_1)
		,
	\end{align*}
	which is fulfilled only if $S^{\alpha\beta}_{\gamma\delta}(\te)=\delta^\alpha_\delta\delta^\beta_\gamma$.

	The proofs of all these statements for $\phi'$ are completely analogous.
\end{proof}

Despite not being point-localized Wightman fields, the fields $\phi,\phi'$ satisfy interesting relative commutation relations which make them useful tools in our construction. Namely, as in the scalar case, it can be shown that these fields are {\em relatively wedge-local}. For a precise formulation of this property, first recall that the {\em right wedge} is the region
\begin{align}\label{eq:WR}
	W_R
	:=
	\{x\in\Rl^2\,:\,x_1>|x_0|\}
	\,,
\end{align}
and the {\em set of all wedges} $\W$ is the orbit of $W_R$ under the natural action of $\PG_+$ on $\Rl^2$. As $W_R$ is invariant under boosts, it consists of all translates of $W_R$ and the {\em left wedge} $W_L:=jW_R=-W_R$, which coincides with the causal complement $W_R'$ of $W_R$.

The field operators $\phi'(f)$, $\phi(g)$ are localized in wedges in the following sense. The assignment of $\phi'(f)$ to the localization region $(W_R+\supp f)''$ (the smallest right wedge containing $\supp f$, where $\supp f$ is defined as the smallest subset of $\Rl^2$ containing $\supp f_\alpha$ for all $\alpha\in\I$), and of $\phi(g)$ to the localization region $(W_L+\supp g)''$, is consistent with the principles of Poincar\'e covariance and Einstein causality. The covariance properties have been established in Proposition \ref{proposition:PhiPhi'} \refitem{item:PhiCovariance} and \refitem{item:PhiPhi'Covariance}. But it remains to prove the locality property, i.e. to check that field operators localized in spacelike separated wedges commute. For this to hold, the analyticity properties of the $S$-matrix, and in particular its crossing symmetry (Def.~\ref{definition:SMatrix} \refitem{item:Crossing}), turn out to be crucial.

\begin{theorem}\label{theorem:PhiPhi'WedgeLocality}
	The fields $\phi$ and $\phi'$ are relatively wedge-local: For any $a\in\Rl^2$, $f\in\Ss(W_R+a)\ot\K$, $g\in\Ss(W_L+a)\ot\K$, and $\Psi\in\DD$, we have
	\begin{align}\label{eq:PhiPhi'Commutator}
		[\phi'(f),\phi(g)]\Psi=0\,.
	\end{align}
\end{theorem}
\begin{proof}
	In view of the translational covariance of $\phi$ and $\phi'$, and the invariance of $\DD$ under translations, it is sufficient to consider the case $a=0$. Taking also into account the strong continuity of $f\mapsto\phi'(f)$ and $g\mapsto\phi(g)$ on $\DD$, we may furthermore restrict to compactly supported smooth $f$ and $g$.

	With the help of Proposition \ref{proposition:ZZ'CommutationRelations}, we first simplify the field commutator on a vector $\Psi_n\in\Hil_n$ of definite particle number $n\in\Nl_0$ to
	\begin{align*}
		[\phi'(f),\phi(g)]\Psi_n
		&=
		[\zd(f^+)'+z(Jf^-)',\,\zd(g^+)+z(Jg^-)]\Psi_n
		\\
		&=
		[\zd(f^+)',\,z(Jg^-)]\Psi_n
		+
		[z(Jf^-)',\,\zd(g^+)]\Psi_n
		\\
		&=
		(L_n^{f^+,Jg^-}+K_n^{Jf^-,g^+})\Psi_n
		\,,
	\end{align*}
	where $L_n^{f^+,Jg^-}$ and $K_n^{Jf^-,g^+}$ act by multiplication with the tensors \eqref{eq:MixedCommutator1} and \eqref{eq:MixedCommutator2}. Recalling the form \eqref{eq:SSigma_nExplicit} of their components, and $(Jf^-)^\gamma(\te')=\overline{f^-_{\overline{\gamma}}(\te')}$, it becomes apparent that for proving $[\phi'(f),\phi(g)]\Psi_n=0$ we have to show that
	  \begin{align}\label{eq:KIntegral}
	       K_n^{Jf^-,g^+}(\bte)^\balpha_\bbeta
	       &=
	       \int d\te'\,f^-_{\overline{\gamma}}(\te')\sum_{\xi_1,...,\xi_{n+1}}\delta^\gamma_{\xi_{n+1}}\delta^\delta_{\xi_1}
	       \prod_{l=1}^n S^{\alpha_l\xi_{l+1}}_{\xi_l\beta_l}(\te'-\te_l)
		\cdot
	       g^+_\delta(\te')
	  \end{align}
	  coincides with
	  \begin{align}\label{eq:LIntegral}
          -L_n^{f^+,Jg^-}(\bte)^\balpha_\bbeta
	  &=
	  \int d\te'\,f^+_\gamma(\te')
	  \sum_{\xi_1,...,\xi_{n+1}}
	  \delta^\gamma_{\xi_{n+1}}\delta^\delta_{\xi_1}
	       \prod_{l=1}^n \overline{S^{\beta_l\xi_{l+1}}_{\xi_l\alpha_l}(\te'-\te_l)}
		\cdot
	       g^-_{\overline{\delta}}(\te')
	  \end{align}
     for all $\bte\in\Rl^n$, $\balpha,\bbeta\in\I^n$. We first make some comments about analyticity properties of the various functions appearing in these integrals.
	In view of the compact support of $f,g$, for arbitrary $\gamma,\delta\in\I$, the functions $\te'\mapsto f^-_{\overline{\gamma}}(\te')$ and $\te'\mapsto g^+_{\delta}(\te')$ continue to entire analytic functions. We also recall from \cite{Lechner:2003} that because the supports of $f,g$ are restricted to wedges, these continuations are bounded on the strip $\Strip(0,\pi)$, and $f^-_{\overline{\gamma}}(\te'+i\mu)$ and $g^+_{\delta}(\te'+i\mu)$ decay rapidly to zero as $\te'\to\pm\infty$, uniformly in $\mu\in[0,\pi]$. Finally, the boundary values at the upper end of the strip are given by $f^-_{\overline{\gamma}}(\te'+i\pi)=f^+_{\overline{\gamma}}(\te')$ and $g^+_{\delta}(\te'+i\pi)=g^-_\delta(\te')$, as can be seen from \eqref{eq:fpm}.

     According to Definition \ref{definition:SMatrix}, also $S$ has a bounded analytic continuation to $\Strip(0,\pi)$, with crossing symmetric boundary value $S^{\alpha\beta}_{\gamma\delta}(i\pi-\te)=S^{\overline{\gamma}\alpha}_{\delta\overline\beta}(\te)$. We may thus shift the contour of integration in \eqref{eq:KIntegral} from $\Rl$ to $\Rl+i\pi$, where it reads
     \begin{align*}
       K_n^{Jf^-,g^+}(\bte)^\balpha_\bbeta
	       &=
	       \int d\te'\,f^+_{\overline{\gamma}}(\te')\sum_{\xi_1,...,\xi_{n+1}}\delta^\gamma_{\xi_{n+1}}\delta^\delta_{\xi_1}
	       \prod_{l=1}^n S^{\overline{\xi_l}\alpha_l}_{\beta_l\overline{\xi_{l+1}}}(\te_l-\te')
		\cdot
	       g^-_\delta(\te')
	       \\
	       &=
	       \int d\te'\,f^+_{\gamma}(\te')\sum_{\xi_1,...,\xi_{n+1}}
	       \delta^{\overline{\gamma}}_{\overline{\xi_{n+1}}}
	       \delta^{\overline{\delta}}_{\overline{\xi_1}}
	       \prod_{l=1}^n
	       S^{\xi_l\alpha_l}_{\beta_l\xi_{l+1}}(\te_l-\te')
		\cdot
	       g^-_{\overline{\delta}}(\te')
	       \,.
     \end{align*}
     But since $S(-\te)=S(\te)^*$, we have $S^{\xi_l\alpha_l}_{\beta_l\xi_{l+1}}(\te_l-\te')=\overline{S_{\xi_l\alpha_l}^{\beta_l\xi_{l+1}}(\te'-\te_l)}$, which proves that the above integral coincides with \eqref{eq:LIntegral}. As $n$ was arbitrary, \eqref{eq:PhiPhi'Commutator} follows.
\end{proof}

The crossing property lying at the basis of this theorem provides a close link between particles and fields \cite{Schroer:2010}. Its importance for the wedge-locality of associated fields was known before in the scalar case \cite{Schroer:1999, Lechner:2003}, and by the results presented here, one sees that this links also persists in the presence of more realistic particle spectra, involving antiparticles and charge conjugation.

By construction, the fields $\phi,\phi'$ generate only single particle states from the vacuum and are solutions of the Klein-Gordon equation: If $f_\alpha=0$ for all $\alpha\in\I$ except some index $\alpha=\alpha_0$, the field $\phi(f)$ solves the Klein-Gordon equation with mass $m_{\alpha_0}$. In addition, these fields are localized in wedges and behave in a continuous and bounded manner under Poincar\'e transformations. Thus they are examples of so-called {\em temperate polarization-free generators} \cite{SchroerWiesbrock:2000-1, BorchersBuchholzSchroer:2001}.

\section{Operator-algebraic formulation and local fields}\label{section:operator-algebras}

So far our construction proceeded from a given $S$-matrix $S\in\SF$ to a vacuum Hilbert space $\Hil$ and a pair $\phi,\phi'$ of Poincar\'e covariant wedge-local fields acting on $\Hil$. Whereas these operators are convenient and useful objects for constructing the model with $S$-matrix $S$, they must not be confused with potentially existing point-like localized quantum fields obeying the equation of motion of the dynamics corresponding to $S$. These {\em local} interacting fields, generically denoted $\varphi$ here, are bound to be much more involved objects whose properties are usually studied by methods like perturbation theory, form factor expansions, or Euclidean methods. In the approach presented here, all fields/observables with sharper than wedge-like localization are derived quantities which will be characterized in terms of the auxiliary fields $\phi,\phi'$. For this, it will be advantageous to formulate our models in an operator-algebraic fashion. As the steps necessary for this reformulation are
almost
identical to the scalar case, which is well documented in the literature \cite{Lechner:2003, BuchholzLechner:2004, Lechner:2006}, we can be brief here.

One first proceeds from the field operators $\phi,\phi'$ to the von Neumann algebras they generate, and introduces, $x\in\Rl^2$,
\begin{align}\label{eq:M}
	\F(W_L+x)
	&:=
	\{e^{i\phi(f)}\,:\,f=f^*\in\Ss(W_L+x)\ot\K\}''
	\,,\\
	\F(W_R+x)
	&:=
	\{e^{i\phi'(f)}\,:\,f=f^*\in\Ss(W_R+x)\ot\K\}''
	\,.
\end{align}
Here the (double) prime denotes the (double) commutant in $\B(\Hil)$, i.e. to any wedge $W\in\W$ we associate a von Neumann algebra $\F(W)\subset\B(\Hil)$ generated by the unitaries $\exp i\phi^{(\prime)}(f)$, $f=f^*\in\Ss(W)\ot\K$. The basic properties of these algebras are collected in the following proposition.

\begin{proposition}\label{proposition:WedgeAlgebras}
	Let $S\in\SF$. The above defined algebras $\F(W)$, $W\in\W$, have the following properties, $W,\Wti\in\W$.
	\begin{propositionlist}
		\item Isotony: $\F(W)\subset\F(\Wti)$ for $W\subset\Wti$,
		\item Covariance: $U(x,\la)\F(W)U(x,\la)^{-1}=\F(\La_\la W+x)$, $(x,\la)\in\PG_+$,
		\item Gauge symmetry: $V(g)\F(W)V(g)^{-1}=\F(W)$, $g\in G$.
		\item Locality: $\F(W)\subset\F(\Wti)'$ for $W\subset\Wti'$,
		\item Cyclicity: The vacuum vector $\Om$ is cyclic and separating for $\F(W)$.
	\end{propositionlist}
\end{proposition}

The proofs of these facts require only trivial changes in the existing proofs for the scalar case, so that we can content ourselves with a few comments: {\em i)} and {\em ii)} are straightforward consequences of the definition of $\F(W)$ and Proposition \ref{proposition:PhiPhi'} \refitem{item:PhiCovariance}. The gauge symmetry {\em iii)} holds because of the transformation law Proposition~\ref{proposition:PhiPhi'}~\refitem{item:PhiGaugeSymmetries} for $\phi$ and $\phi'$. To show {\em iv)}, one has to check that the field commutation relations of Theorem~\ref{theorem:PhiPhi'WedgeLocality} hold also in the stronger sense that the associated unitary groups commute,
\begin{align}\label{eq:UnitaryGroupsCommute}
	[e^{i\phi'(f)},\,e^{i\phi(g)}]=0
	\,,\qquad
	f=f^*\in\Ss(W_R)\ot\K,\,g=g^*\in\Ss(W_L)\ot\K\,.
\end{align}
This can be done by a calculation on analytic vectors, as in the free field case (\cite{ReedSimon:1975}, see also \cite[Prop.~5.1]{Lechner:2011}). Actually, {\em wedge duality} holds, a property stronger than locality: For each $W\in\W$, one has
\begin{align*}
	 \F(W)'=\F(W')\,.
\end{align*}
The cyclicity statement in {\em v)} is a consequence of Proposition \ref{proposition:PhiPhi'} \refitem{item:PhiReehSchlieder}, and thanks to locality, this implies that $\Om$ is separating for each $\F(W)$ as well.
\\
\\
As mentioned in the previous section, the unbounded field operators $\phi(f)$ and $\phi'(f)$ play the role of polarization-free generators in this context, and are affiliated to the field algebras $\F((W_L+\supp f)'')$ and $\F((W_R+\supp f)'')$ respectively. In particular, by choosing $f$ to have non-zero components only in the representation space $\K_q$ of the irreducible representation $V_{1,q}$ of charge $q$, we obtain field operators which are localized in wedges (equal to spacelike cones in two dimensions) and interpolate between the vacuum and the single particle states in the sector $q$. These results fit well into the model-independent framework of Buchholz and Fredenhagen \cite{BuchholzFredenhagen:1982}.

In the context of such a system of field algebras associated with wedge regions, one can unambiguously define the maximal algebras of fields $\F(\OO)$ localized in smaller spacetime regions $\OO\subset\Rl^2$ as follows \cite{Borchers:1992, BuchholzLechner:2004}: For a double cone, that is a region of the form $\OO_{xy}:=(W_L+x)\cap(W_R+y)$, $x-y\in W_R$, one puts
\begin{align}\label{eq:F(O)}
	\F(\OO_{xy})
	:=
	\F(W_L+x)\cap\F(W_R+y)
	\,.
\end{align}
Algebras associated with arbitrary regions $\OO\subset\Rl^2$ can then be defined by additivity. This construction results in a local net $\OO\mapsto\F(\OO)$ of von Neumann algebras $\F(\OO)\subset\B(\Hil)$ indexed by the family of all subsets $\OO\subset\Rl^2$, and it can be readily checked that this net inherits the basic features isotony, Poincar\'e covariance, gauge invariance and locality from the corresponding properties of the wedge algebras (Proposition \ref{proposition:WedgeAlgebras} {\em i)--iv)}).

The algebras $\F(\OO)$, where $\OO$ is a bounded localization region, can be thought of as being generated by (bounded functions of) local quantum fields $\varphi$ underlying the model, smeared with test functions having support in $\OO$.
In the approach followed here, these local fields are not constructed explicitly\footnote{See however \cite{BostelmannCadamuro:2012} for recent progress in this direction.}, but rather characterized indirectly as elements of the algebra intersections \eqref{eq:F(O)}. It is therefore not clear if such local operators exist at all, or if the model is trivial in the sense that $\F(\OO)$ consists only of multiples of the identity for bounded regions $\OO$. At least three different scenarios regarding the ``size'' of the local field algebra $\F(\OO)$, where $\OO$ is bounded, are conceivable:
\begin{enumerate}
	\item[1)] $\F(\OO)$ has the vacuum $\Om$ as a cyclic vector and is in particular non-trivial.
	\item[2)] $\F(\OO)$ is non-trivial, but does not have the vacuum $\Om$ as a cyclic vector.
	\item[3)] $\F(\OO)$ is trivial, i.e. $\F(\OO)=\Cl\cdot1$.
\end{enumerate}
Scenario 1) describes the situation encountered in a local quantum field theory with localizable charges \cite{DoplicherHaagRoberts:1971,DoplicherHaagRoberts:1974}, whereas Scenario 2) occurs in local theories with gauge charges \cite{BuchholzFredenhagen:1982}. In this case, one would expect that the {\em observable} algebras $\A(\OO)$, i.e. the gauge invariant subalgebras of $\F(\OO)$, have the vacuum as a cyclic vector on the charge zero subspace. Scenario 3) however does not occur in any local field theory, as it implies that this theory does not have any observables localized in $\OO$, and should thus be considered a pathology.

For particular choices of $S$, like the $S$-matrices corresponding to the $O(N)$ $\sigma$-models discussed in Section \ref{section:examples}, compelling evidence exists from other approaches (perturbation theory, large $N$ limits, Euclidean methods, lattice constructions) to the effect that local quantum fields {\em do} exist, and hence these models do not belong to Scenario 3). Rigorous proofs are however usually difficult to obtain, and in the present generality, it seems well possible to also build pathological $S$-matrices which can not be associated with a local theory. Therefore the question arises for which $S\in\SF$ there exists a local quantum field theory with the $S$-matrix given by $S$ as its scattering operator.

The question whether the local algebras \eqref{eq:F(O)} are non-trivial is precisely the question if there exists a local quantum field theory with the $S$-matrix given by $S$ as its scattering operator. After this existence question has been settled, one could look into extending the model-independent analysis of Bostelmann \cite{Bostelmann:2005} to models in two dimensions, and possibly also construct such associated local field operators more explicitly.

Fortunately, a clear-cut sufficient criterion for Scenario~1) exists in the operator-algebraic framework of quantum field theory\footnote{For the gauge-invariant subalgebras on the charge zero space, it can also be applied in Scenario 2).}. This method, known as the modular nuclearity condition, involves the modular operator $\Delta$ of $(\F(W_R),\Om)$, which exists since $\Om$ is cyclic and separating for $\F(W_R)$ \cite{Takesaki:2003}, and the maps, $x\in W_R$,
\begin{align}\label{eq:Xi}
	\Xi(x):\F(W_R)\to\Hil
	\,,\qquad
	\Xi(x)(F)
	:=
	\Delta^{1/4}U(x)F\Om\,.
\end{align}
Using modular theory, it is not difficult to see that $\Xi(x)$ is bounded as a map between the Banach spaces $(\F(W_R),\|\cdot\|_{\B(\Hil)})$ and $\Hil$. In analyzing the size of the intersection
\eqref{eq:F(O)}, it is a crucial question if $\Xi(x)$ is {\em nuclear}, too, i.e. if it can be decomposed into a series of rank one maps with summable norms. Building on earlier investigations of
nuclearity properties \cite{BuchholzWichmann:1986, BuchholzDAntoniLongo:1990-1} and the closely related split property \cite{DoplicherLongo:1984, Mueger:1998}, the following two results have been
established in \cite{BuchholzLechner:2004} and \cite{Lechner:2008}, respectively.

\begin{theorem}
	Assume that the $\Xi(x)$ are nuclear for $x\in W_R$. Then, for any double cone $\OO$,
	\begin{propositionlist}
		\item $\F(\OO)$ is isomorphic to the hyperfinite type III$_1$ factor.
		\item The vacuum vector $\Om$ is cyclic for $\OO$.
	\end{propositionlist}
\end{theorem}

\noindent For further consequences of the nuclearity condition, see \cite{Mueger:1998, Lechner:2008}.

Stated informally, this theorem says that in case the $\Xi(x)$ are nuclear, local operators exist in abundance, as expected in a well-behaved local quantum field theory. Moreover, they can be used to construct scattering states, since in particular local operators interpolating between the vacuum and the single particle space (with definite charge) exist. These matters will be discussed in the next section.

Verifying the modular nuclearity condition is thus a possible strategy for proving the existence of a local quantum field theory with the considered $S$-matrix without having to deal with the explicit construction of local interacting field operators. Despite its abstract formulation, this condition takes a rather concrete form in the models at hand. Namely, the modular operator $\Delta$ appearing in \eqref{eq:Xi} acts as a (imaginary) boost transformation. This fact was proven for the scalar case in \cite{BuchholzLechner:2004}, and can be generally expected from the Bisognano-Wichmann theorem \cite{BisognanoWichmann:1976, Mund:2001}. Concretely, we have
\begin{align}
	\Delta^{it}=U(0,-2\pi t)
	\,,\qquad
	t\in\Rl\,.
\end{align}

Again we refrain from giving a formal proof and only indicate the necessary argument. One basically proceeds as in \cite[Prop.~3.1]{BuchholzLechner:2004}, with the only difference that the Poincar\'e group does not act irreducibly on the single particle space here because of the richer particle spectrum considered. But the ``correct'' action of the modular group on the single particle space can be explicitly computed because the fields $\phi,\phi'$ generate single particle states from the vacuum, and these are identical to the ones in free field theory as the $S$-matrix does not enter on the single particle level.

In view of this geometric form of the modular operator, the $\Delta^{1/4}$ in \eqref{eq:Xi} corresponds to an imaginary boost by $\frac{i\pi}{2}$ in the center of mass rapidity. In terms of analytic continuation, that means
\begin{align*}
	(\Xi(x)(F))_n^\balpha(\bte)
	&=
	(U(x)F\Om)_n^\balpha(\te_1-\tfrac{i\pi}{2},...,\te_n-\tfrac{i\pi}{2})
	\\
	&=
	\prod_{k=1}^n e^{m_{\alpha_k}(x_0\sinh\te_k-x_1\cosh\te_k)}
	\cdot
	(F\Om)_n^\balpha(\te_1-\tfrac{i\pi}{2},...,\te_n-\tfrac{i\pi}{2})
	\,.
\end{align*}
Nuclearity estimates of the map $\Xi(x)$ can now be established by showing that the space of the functions $(\Xi(x)(F))_n$ is ``small'' in a specific sense. Two ingredients have to be taken into account here: First, the rapidly decreasing factors $e^{m_{\alpha_k}(x_0\sinh\te_k-x_1\cosh\te_k)}$ (rapidly decreasing in $\te_k$ because $x\in W_R$ and $m_{\alpha_k}>0$), and second, analytic properties of the functions $(F\Om)_n$, $F\in\F(W_R)$. These momentum space analyticity properties derive on the one hand from the spacetime localization of $F$ in the wedge $W_R$, and on the other hand from analyticity properties of the $S$-matrix which enters via the symmetry properties of $(F\Om)_n$. For scalar $S$-matrices, this analysis has been carried out in \cite{Lechner:2008}, and a proof of the modular nuclearity condition for a class of so-called regular (scalar) $S$-matrices has been given.

In the case of general $S$-matrices, the multi-component nature of the $n$-particle functions $(F\Om)_n\in L^2(\Rl^n,d\bte)\ot\K^{\ot n}$ requires a somewhat more involved analysis. However, the indicated
strategy seems to be applicable also here and looks promising in principle \cite{Alazzawi:2013}. We will not enter the discussion of the modular nuclearity condition in the present article in detail, but rather
present some evidence towards its validity in the context of specific $S$-matrices in Section~\ref{section:examples}.

\section{Scattering states and reconstruction of the ${S}$-matrix}\label{section:scattering}

Up to now, the underlying $S$-matrix entered our construction via the symmetrization properties of the vacuum Hilbert space, and the commutation relations of the creation/annihilation operators on this space. In this section, we will explain the physical significance of $S$ by establishing its close connection to the scattering operator of the constructed model.

We will thus be concerned with the calculation of scattering states, and employ the usual methods of Haag-Ruelle-Hepp scattering theory \cite{Araki:1999, Hepp:1965}, taking into account the wedge-locality of the fields as in \cite{BorchersBuchholzSchroer:2001} and the charge structure as in \cite{DoplicherHaagRoberts:1974}. It is a basic prerequisite for the construction of multi particle scattering states that quasilocal operators interpolating between the vacuum and the single particle space exist, and we will therefore require throughout this section --- somewhat stronger than necessary --- that the vacuum vector $\Om$ is cyclic for the field algebra $\F(\OO)$ of some double cone $\OO$, i.e. that Scenario~1) of the previous section applies. As explained there, this assumption is in particular satisfied if the maps $\Xi(x)$ \eqref{eq:Xi} are nuclear.

It is well known that in case the vacuum is cyclic for some field algebra $\F(\OO)$, this algebra also contains field operators $F_q$ of definite charge $q\in\Q$ \cite{DoplicherHaagRoberts:1969}. In particular, picking a test function $h\in\Ss(\Rl^2)$ such that $\supp\hti$ intersects the energy momentum spectrum in the sector $q$ only in $H^+_{m(q)}$ --- recall that according to our assumptions, there exists precisely one isolated mass shell in this sector\footnote{In the case of embedded mass eigenvalues, one would need to employ the methods of \cite{Dybalski:2005} to calculate scattering states.} --- the quasi-local operator $F_q(h)=\int dx\,h(x)\,U(x,0)F_qU(x,0)^{-1}$ creates a single particle state of charge $q$ from the vacuum, that is, $F_q(h)\Om\in\Hil_{1,q}$.

We also introduce the velocity support of mass $m$ of a function $h\in\Ss(\Rl^2)$ as
\begin{align}
	\VV_m(h)
	:=
	\{(1, p_1(p_1^2+m^2)^{-1/2})\,:\,(p_0,p_1)\in\supp\hti\}
	\,,
\end{align}
and for vectors $\Psi_{1,q}\in\Hil_{1,q}$, the velocity support $\VV(\Psi_{1,q})$ is defined as the same set, with mass $m=m(q)$ and the energy momentum spectral support of $\Psi_{1,q}$ instead of $\supp\hti$. It is a consequence of the cyclicity of $\Om$ for $\F(\OO)$ that there exist sufficiently many quasi local creation operators: Given $\Psi_{1,q}\in\Hil_{1,q}$ and $\eps>0$, we find $F\in\F(\OO)$ and $h\in\Ss(\Rl^2)$, with $h$ having velocity support in an arbitrarily small neighborhood of $\VV(\Psi_{1,q})$, such that $\|F_q(h)\Om-\Psi_{1,q}\|<\eps$.

These quasi-local creation operators are related to the asymptotic creation operators as follows. For $h\in\Ss(\Rl^2)$, we define $h_{t,m}\in\Ss(\Rl^2)$, $t\in\Rl$, $m>0$, by
\begin{align}\label{eq:KleinGordon}
	\hti_{t,m}(p)
	:=
	e^{i(p_0-(p_1^2+m^2)^{1/2})\,t}\cdot\hti(p)
	\,,
\end{align}
so that the $t$-dependence of this function drops out on the mass shell $H^+_m$. That is, $h_{t,m}^+=h^+$ and $F_q(h_{t,m(q)})\Om$ is a vector in $\Hil_{1,q}$ independent of $t$. Recall the support properties of $h_{t,m}$ for $t\to\pm\infty$  \cite{Hepp:1965}: Given any $\eps$-neighborhood $\VV_m^\eps(h)$ of $\VV_m(h)$, there exists a Schwartz function $\hhat$ with support in $t\,\VV_m^\eps(h)$, such that for any $N$, we have $|t|^N(\hhat_{t,m}-h_{t,m})\to0$ in the topology of $\Ss(\Rl^2)$ as $t\to\pm\infty$.

The basic statement of Haag-Ruelle scattering theory then is \cite{Araki:1999,DoplicherHaagRoberts:1974,BuchholzFredenhagen:1982} that given field operators $F_1,...,F_n\in\F(\OO)$, charges $q_1,...,q_n\in\Q$, testfunctions $h_1,...,h_n$ with disjoint velocity supports $\VV_{m(q_1)}(h_1),...,\VV_{m(q_n)}(h_n)$, the limits
\begin{align}
	\lim_{t\to\pm\infty}
	F_{1,q_1}(h_{1,t,m(q_1)})\cdots F_{n,q_n}(h_{n,t,m(q_n)})\Om
	=:
	(\psi_1\times...\times\psi_n)_{\rm out/in}
\end{align}
exist and and only depend on the single particle vectors $\psi_k:=F_{k,q_k}(h_k)\Om$. Furthermore, the dependence of $(\psi_1\times...\times\psi_n)_{\rm out/in}$ on the $\psi_k$ is linear and continuous.

Because our auxiliary fields $\phi,\phi'$ are only wedge-local, we will need a somewhat refined analysis, similar to the arguments presented in \cite{BuchholzFredenhagen:1982, BorchersBuchholzSchroer:2001, Lechner:2008}. As in \cite{BorchersBuchholzSchroer:2001}, we write $h\prec_m h'$ for two testfunctions $h,h'\in\Ss(\Rl^2)$ with velocity supports ordered such that $\VV_m(h')-\VV_m(h)\subset\{0\}\times\Rl_+$.

All these notations will also be used for multi component single particle functions of fixed charge and mass. Test functions $f\in\Ss(\Rl^2)\ot\K= \bigoplus_{q\in\Q}\Ss(\Rl^2)\ot\K_q$ will be decomposed according to $f=\bigoplus_q f_q$, and we write $f_{q,t}$ to denote the function with $t$-dependence of every component as in \eqref{eq:KleinGordon} and mass $m=m(q)$. The velocity support $\VV(f_q)$ will be understood as the union of the velocity supports (with mass $m(q)$) of all the components of $f_q$.

Having recalled these facts, we come to the calculation of scattering states in the models at hand.

\begin{proposition}\label{Proposition:ScatteringStates}
	Let $f_1,...,f_n\in\Ss(\Rl^2)\ot\K$ be test functions with $\supp \fti_k$ contained in the forward light cone, and ordered velocity supports, $f_1\prec ... \prec f_n$. Then
	\begin{align}
		\left(f_1^+\times ... \times f_n^+\right)_{\rm out}
		&=
		\phi(f_1)\cdots\phi(f_n)\Om
		=
		\sqrt{n!}P_n(f_1^+\ot...\ot f_n^+)
		\label{eq:ScatteringStatesOut}
		\,,\\
		\left(f_1^+\times ... \times f_n^+\right)_{\rm in}
		&=
		\phi(f_n)\cdots\phi(f_1)\Om
		=
		\sqrt{n!}P_n(f_n^+\ot...\ot f_1^+)
		\,.
	\end{align}
\end{proposition}
\begin{proof}
	The proof follows closely the one of Lemma~6.1 in \cite{Lechner:2008} and proceeds by induction in $n$. We first consider $n=1$ and pick some charge $q\in\Q$. Then $\phi(f_ {1,q})\Om=f_{1,q}^+\in\Hil_{1,q}$ and $f_{1,q}^+=(f_{1,q}^+)_{\rm in}=(f_{1,q}^+)_{\rm out}$. Summing over $q\in\Q$, we find $\phi(f_1)\Om=f_1^+=(f_1^+)_{\rm out/in}$.

	For the induction step $n\to n+1$, we only consider the limit $t\to\infty$; the case $t\to-\infty$ is analogous. Fix charges $q,q_1,...,q_n$, and test functions  $f,f_1,...,f_n\in\Ss(\Rl^2)\ot\K$ with compact supports in momentum space around points on the mass shells with masses $m(q), m(q_1),...,m(q_n)$ such that $f\prec f_1\prec...\prec f_n$. In a first step, given $\eps>0$ we find field operators $F_1,...,F_n\in\F(\OO)$ such that $\|\psi_k-f_{k,q_k}^+\|<\eps$, where $\psi_k:=F_{k,q_k}(h_{k})\Om\in\Hil_{1,q_k}$.

	Note that the operators $G_{k,t}:=F_{k,q_k}(\hat{h}_{k,t,m(q_k)})$ satisfy $G_{1,t}\cdots G_{n,t}\Om\to(\psi_1\times...\times\psi_n)_{\rm out}$ as $t\to\infty$. This is the case because $F_{1,q_1}(h_{1,t,m(q_1)})\cdots F_{n,q_n}(h_{n,t,m(q_n)})\Om\to(\psi_1\times...\times\psi_n)_{\rm out}$ and $\hat{h}_{k,t,m(q_k)}-{h}_{k,t,m(q_k)}\to0$ rapidly in the topology of $\Ss(\Rl^2)$, whereas $\|h_{k,t,m(q_k)}\|_1$ is polynomially bounded in $t$.

	The operators $G_{k,t}$ are localized in $\OO+t\VV_{m(q_k)}(h_k)$, and the field $\phi(\fhat_{t,q})$ is localized in $W_L+t\VV_{m(q)}(f_q)$. Because of the ordering $f\prec f_1\prec ...\prec f_n$, these two regions are spacelike separated for sufficiently large $t$, and the $G_{k,t}$ commute with $\phi(\fhat_{t,q})$ on scattering states.

	To establish the claim, we now pick an arbitrary vector $\Psi\in\DD$ and note that the field operator $\phi(f_q)$ satisfies $\lim_{t\to\infty}\phi(\fhat_{t,q})^*\Psi=\phi(f_q)^*\Psi$; again because $\fhat_{t,q}-f_{t,q}\to0$ in $\Ss(\Rl^2)\ot\K_q$, and $\phi(f_{t,q})^*\Psi$ is independent of $t$. Thus we can compute
	\begin{align*}
		\langle\phi(f_q)^*\Psi,\,(\psi_1\times...\times\psi_n)_{\rm out}\rangle
		&=
		\lim_{t\to\infty}
		\langle\phi(\fhat_{t,q})^*\Psi,\,
		G_{1,t}\cdots G_{n,t}\Om\rangle
		\\
		&=
		\lim_{t\to\infty}
		\langle\Psi,\,
		G_{1,t}\cdots G_{n,t}\phi(\fhat_{t,q})\Om\rangle
		\\
		&=
		\lim_{t\to\infty}
		\langle\Psi,\,
		G_{1,t}\cdots G_{n,t} f_q^+\rangle
		\\
		&=
		\langle\Psi,\,
		(\psi_1\times...\times \psi_n\times f_q^+)_{\rm out}\rangle
		\,,
	\end{align*}
	where in the last step, we used the convergence of the $G_{k,t}$ to the asymptotic creation operators. Taking into account that $\DD\subset\Hil$ is dense, and the Bose symmetry of the scattering states, we arrive at
	\begin{align*}
		\phi(f_q)(\psi_1\times...\times\psi_n)_{\rm out}
		&=
		( f_q^+\times\psi_1\times...\times \psi_n)_{\rm out}
		\,.
	\end{align*}
	Proceeding from the $\psi_k$ to the $f_{k,q_k}^+$, we find in view of the continuous dependence of multi-particle scattering states on their single particle components, and the induction hypothesis
	\begin{align*}
		\phi(f_q)\phi(f_{1,q_1})\cdots\phi(f_{n,q_n})\Om
		=
		\phi(f_q)(f_{1,q_1}^+\times...\times f_{n,q_n}^+)_{\rm out}
		&=
		(f_q^+\times f_{1,q_1}^+\times...\times f_{n,q_n}^+)_{\rm out}
		\,.
	\end{align*}
	Taking linear combinations over $q,q_k$ gives the claimed result. The second equation in \eqref{eq:ScatteringStatesOut} holds by definition of $\phi$ and the fact that only the creation parts of the fields contribute to this vector because of $f_k^-=0$, $k=1,\ldots,n$.
\end{proof}

\begin{proposition}
		The sets of incoming and outgoing $n$-particle collision states constructed in Proposition~\ref{Proposition:ScatteringStates} are total sets in $\Hil_n$, i.e., the model is asymptotically complete.
\end{proposition}
\begin{proof}
	When $f\in\Ss(\Rl^2)$ varies over all functions whose Fourier transforms have compact support within the forward light cone, $f^+$ ranges over a dense set in $\Hil_1$. So the statement is true for $n=1$. For higher $n$, the ordering $f_1\prec...\prec f_n$ has to be taken into account. 	As the asymptotic states \eqref{eq:ScatteringStatesOut} depend only on the mass shell restrictions of the $f_1,...,f_n$, we can change these functions off the mass shell in such a way that $\VV_m(f_k)=\{(1,p_1(p_1^2+m^2)^{-1/2})\,:\,(p_0,p_1)\in\supp\fti_k\cap H^+_m\}$, without changing the scattering states. Parametrizing the mass shell $H^+_m$ by the rapidity according to $p_m(\te)=(m\ch\te,m\sh\te)$ then shows $\VV_{m_{[\alpha]}}(f_{k,\alpha})=\{(1,\tanh\te)\,:\,\te\in\supp f_{k,\alpha}^+\}$. But $\tanh$ is a strictly monotonously increasing function, and thus $f_1\prec f_2$ is equivalent to $\supp f_2^+-\supp f_1^+\subset\Rl^+$. Hence the $f_1^+\ot...\ot f_n^+$ span a dense set in $L^2(E_n)\ot\K^{\ot n}$, where $E_n:=\{(\
te_1,...,\te_n)\in\Rl^n\,:\,\te_1\leq...\leq\te_n\}$. But when $L^2(E_n)\ot\K^{\ot n}\subset L^2(\Rl^n)\ot\K^{\ot n}$ is considered as a subspace by continuing the functions on $E_n$ by zero to $\Rl^n$, the $S$-symmetrization projection $P_n:L^2(E_n)\ot\K^{\ot n}\to P_n(L^2(\Rl^n)\ot\K^{\ot n})=\Hil_n$ is a continuous map and onto. This shows that the $n$-particle collision states form a dense set in the $n$-particle space, and since $n$ was arbitrary, asymptotic completeness follows.
\end{proof}

Having determined the form of the collision states, we will now compute the S-matrix ${\rm S}$, considered as an operator on the totally symmetrized Bose Fock space $\Hil^+=\bigoplus_{n=0}^\infty\Hil_n^+$ over $\Hil_1$. According to our above construction of scattering states, the M\o ller operators $W_{\rm in/out}:\Hil^+\to\Hil$ have the form
\begin{align}
	\label{eq:MollerOut}
	W_{\rm out} P_n^+(f_1^+\ot...\ot f_n^+)
	&=
	P_n(f_1^+\ot...\ot f_n^+)
	\,,\\
	W_{\rm in} P_n^+(f_n^+\ot...\ot f_1^+)
	&=
	P_n(f_n^+\ot...\ot f_1^+)
	\,,
	\label{eq:MollerIn}
\end{align}
where $f_1\prec...\prec f_n$ and $P_n^+$ denotes the total symmetrization, given by $S=1$. These are well-defined linear operators with dense domains and ranges which extend to unitaries since the norms of $P_n^+(f_1^+\ot...\ot f_n^+)$ and $P_n(f_1^+\ot...\ot f_n^+)$ coincide as a consequence of the ordering of the supports of the $f_k^+$. The S-matrix is the product of these M\o ller operators,
\begin{align}\label{eq:SMatrix}
	{\rm S}:=W_{\rm out}^*W_{\rm in}:\Hil^+\to\Hil^+\,.
\end{align}

\begin{theorem}\label{Theorem:SMatrix}
	Assume that the vacuum vector $\Om$ is cyclic for the field algebra $\F(\OO)$ for some double cone $\OO$. Then the constructed model solves the inverse scattering problem for the $S$-matrix, i.e. its scattering operator ${\rm S}$ reproduces $S$. More precisely, ${\rm S}$ \eqref{eq:SMatrix} acts as
	\begin{align}
		({\rm S}\Psi^+)_n(\bte)
		=
		{\rm S}_n(\bte)\Psi^+_n(\bte)
		\,,\qquad
		\Psi^+\in\Hil^+\,,
	\end{align}
	and the tensors ${\rm S}_n$ are given by
	\begin{align}\label{eq:Sn}
		{\rm S}_n(\bte)^{\alpha_1...\alpha_n}_{\beta_1...\beta_n}
		&=
		\{S_n^\iota(\te_{\pi(1)},...,\te_{\pi(n)})^{\alpha_{\pi(1)}...\alpha_{\pi(n)}}_{\beta_{\pi(n)}...\beta_{\pi(1)}}\,:\,\te_{\pi(1)}\leq...\leq\te_{\pi(n)}\}
		\,,
	\end{align}
	where $\iota\in\frS_n$ is the inversion permutation $\iota(k):=n+1-k$ and $S_n^\iota$ is defined in \eqref{eq:SPi}. Explicitly, for $n=2$,
	\begin{align}\label{eq:S2}
		{\rm S}_2(\te_1,\te_2)^{\alpha_1\alpha_2}_{\beta_1\beta_2}
		&=
		\left\{
		\begin{array}{rcl}
			S^{\alpha_1\alpha_2}_{\beta_2\beta_1}(\te_2-\te_1) &;& \te_1\leq\te_2
			\\
			S^{\alpha_2\alpha_1}_{\beta_1\beta_2}(\te_1-\te_2) &;& \te_2<\te_1
		\end{array}
		\right.\,.
	\end{align}
\end{theorem}
\begin{proof}
	Let $\bte\in\Rl^n$. Then there exists a permutation $\pi\in\frS_n$ such that $\te_{\pi(1)}\leq...\leq\te_{\pi(n)}$. In view of the ordering $f_1\prec...\prec f_n$, we have
	\begin{align}
		P_n(f_1^+\ot...\ot f_n^+)(\bte)
		&=
		\frac{1}{n!}S_n^\pi(\bte)\,f_1^+(\te_{\pi(1)})\ot...\ot f_n^+(\te_{\pi(n)})
		\,,\\
		P_n^+(f_1^+\ot...\ot f_n^+)(\bte)
		&=
		\frac{1}{n!}F_n^\pi\,f_1^+(\te_{\pi(1)})\ot...\ot f_n^+(\te_{\pi(n)})\,,
	\end{align}
	with the $\theta$-independent $S=F$ (flip) in the second line. Thus $W_{\rm out}$ \eqref{eq:MollerOut} takes the form
	\begin{align*}
		(W_{\rm out}\Psi^+)_n(\bte)
		=
		W_{{\rm out},n}(\bte)\Psi_n^+(\bte)
		\,,\qquad
		W_{{\rm out},n}(\bte)=\{S_n^\pi(\bte)(F_n^\pi)^{-1}\,:\,\te_{\pi(1)}\leq...\leq\te_{\pi(n)}\}
		\,.
	\end{align*}
	Similarly,
	\begin{align*}
		(W_{\rm in}\Psi^+)_n(\bte)
		=
		W_{{\rm in},n}(\bte)\Psi_n^+(\bte)
		\,,\qquad
		W_{{\rm in},n}(\bte)=\{S_n^{\pi\iota}(\bte)(F_n^{\pi\iota})^{-1}\,:\,\te_{\pi(1)}\leq...\leq\te_{\pi(n)}\}
		\,,
	\end{align*}
	where $\iota(k):=n+1-k$ is the total inversion permutation. This implies that ${\rm S}$ acts as
	\begin{align*}
		({\rm S}\Psi^+)_n(\bte)
		=
		{\rm S}_n(\bte)\Psi_n^+(\bte)
		\,,\qquad
		{\rm S}_n(\bte)
		=
		\{F_n^\pi S_n^\pi(\bte)^{-1}S_n^{\pi\iota}(\bte)(F_n^{\pi\iota})^{-1}\,:\,\te_{\pi(1)}\leq...\leq\te_{\pi(n)}\}
		\,.
	\end{align*}
	The $F$-tensors have the components $(F_n^\pi)^\balpha_\bbeta=\delta^{\alpha_{\pi(1)}}_{\beta_1}\cdots\delta^{\alpha_{\pi(n)}}_{\beta_n}$ and form a representation of $\frS_n$. Furthermore, since $D_n$ is a representation of $\frS_n$, one has
	\begin{align*}
		S_n^{\pi\iota}(\bte)=S_n^\pi(\bte)S_n^\iota(\te_{\pi(1)},...,\te_{\pi(n)})
		\,.
	\end{align*}
	Combining these two equations with the above formula for ${\rm S}_n$ gives \eqref{eq:Sn}. For $n=2$, we have $\iota=\tau_1$ and $S_2^\iota(\te_1,\te_2)=S(\te_2-\te_1)$, which gives \eqref{eq:S2}.
\end{proof}

Explicitly, $S_n^\iota$ is a product of $\frac{1}{2}n(n-1)$ factors of $S$, corresponding to $\frac{1}{2}n(n-1)$ consecutive two-body collisions in an $n\to n$ process. As there is no particle production, Theorem \ref{Theorem:SMatrix} shows that the constructed model has the factorizing S-matrix ${\rm S}$ based on $S$ as its scattering operator (provided it contains local observables). Often times the S-matrix is also expressed by scalar products between improper asymptotic states of sharp rapidity. We note that in case of a parity invariant $S$, i.e. $S^{\alpha\beta}_{\gamma\delta}(\te)=S^{\beta\alpha}_{\delta\gamma}(\te)$, we obtain the more familiar formula (see for example \cite{AbdallaAbdallaRothe:1991})
\begin{align*}
	{}_{\rm out}\langle \te_1,\alpha_1;\te_2,\alpha_2|\te_1',\beta_1;\te_2',\beta_2\rangle_{\rm in}
	&=
	\delta(\te_1-\te_1')\delta(\te_2-\te_2')\cdot S^{\alpha_1\alpha_2}_{\beta_2\beta_1}(|\te_1-\te_2|)
	\\
	&\quad+
	\delta(\te_1-\te_2')\delta(\te_2-\te_1')\cdot S^{\alpha_1\alpha_2}_{\beta_1\beta_2}(|\te_1-\te_2|)
	\,.
\end{align*}

\section{Examples of ${S}$-matrices}\label{section:examples}

The construction presented so far was based on an arbitrary $S$-matrix satisfying the assumptions collected in Definition~\ref{definition:SMatrix}; an explicit form of $S$ was not needed. In this section, we will complement the abstract analysis by providing some concrete examples of $S$-matrices.

In the approach taken here, any $S$-matrix defines a model. In the Lagrangian approach to quantum field theory, on the other hand, a model is specified in terms of Lagrangian. A connection between the two approaches can be made whenever the exact S-matrix of some integrable model is available in the Lagrangian setting. This is the case for many models, where $S$ can be obtained by exploiting conservation laws which are assumed to also be present in the quantum theory, comparison with perturbative results, and analyticity assumptions, see  \cite{AbdallaAbdallaRothe:1991,Mussardo:1992,Dorey:1998} and the references cited therein. Below we will see examples of such models.

The simplest class of $S$-matrices is the scalar one, where $\K=\Cl$ and the mass spectrum consists of just a single mass $m>0$. This is the setting of theories containing only a single species of neutral massive particles, and has previously been worked out in \cite{Lechner:2003}. In this case the constraints on $S$ imposed by Definition~\ref{definition:SMatrix} simplify drastically, and it is possible to derive the most general form of $S\in\SF$ explicitly \cite{Lechner:2006}. One finds that the --- here scalar-valued --- $S$-matrix takes the form
\begin{align}\label{eq:ScalarS}
	S(\te)
	=
	\eps\, e^{ia \sinh\te}\,\prod_k\frac{\sinh\beta_k -\sinh\te}{\sinh\beta_k +\sinh\te}
	\,,
\end{align}
where $\eps=\pm1$, $a\geq0$, and the $\beta_k$ form finite or infinite sequences of complex numbers with $0<{\rm Im}\beta_k \leq \frac{\pi}{2}$, subject to certain symmetry and summability conditions \cite[Prop.~3.2.2]{Lechner:2006} which imply the properties Def.~\ref{definition:SMatrix} \refitem{item:Unitarity}, \refitem{item:HermitianAnalyticity}, \refitem{item:Crossing}, and convergence of the product. In the scalar case, the $S$-matrix is also referred to as {\em scattering function}.

This class of scattering functions $S$-matrices contains in particular the function
\begin{align*}
	S_{{\rm ShG}(g)}(\te)
	=
	\frac{\sinh\te-i\sin\frac{\pi g^2}{4\pi +g^2}}{\sinh\te+i\sin\frac{\pi g^2}{4\pi +g^2}}
	\,,
\end{align*}
where $g$ is a real parameter. This function is expected to be the exact scattering function of the Sinh-Gordon model with coupling $g$  \cite{ArinshteinFateevZamolodchikov:1979, BradenSasaki:1991}. It belongs to the subset of {\em regular} scattering functions, defined as the ones with $a=0$ and finite sequences $\{\beta_k\}$. For such $S$, it is also known that the modular nuclearity condition holds, and hence the vacuum vector is cyclic for double cone algebras \cite{Lechner:2008}. Thus the assumptions about the modular nuclearity condition made in Sections \ref{section:operator-algebras} and \ref{section:scattering} are satisfied, and the full S-matrix can be computed as
\begin{align*}
	({\rm S}\Psi^+)_n(\bte)
	=
	\prod_{1\leq l<r\leq n}S(|\te_l-\te_r|)\cdot\Psi_n^+(\bte)
	\,,\qquad
	\Psi^+\in\Hil^+\,.
\end{align*}

In the matrix-valued case with $\dim\K>1$, a simple class of $S$-matrices are so-called diagonal solutions (see also \cite{Jimbo:1986, LiguoriMintchev:1995} for similar $S$-matrices arising in the context of Toda systems). In these examples, one considers a spectrum of $N$ neutral particles of the same mass, that is, puts $\K=\Cl^N$ with some $N\in\Nl$, conjugation $\overline\alpha=\alpha$, and masses $m_\alpha=m$, $\alpha\in\{1,...,N\}$. The $S$-matrix is defined as
\begin{align}\label{eq:DiagonalS}
	S(\te)^{\alpha\beta}_{\gamma\delta}
	:=
	\sigma_{\alpha\beta}(\te)\delta^\alpha_\delta \delta^\beta_\gamma
	\,,
\end{align}
(no sum over $\alpha,\beta$), and thus $S(\te)^{\alpha\beta}_{\gamma\delta}=S(\te)^{\beta\alpha}_{\gamma\delta}$ can be regarded as a diagonal $(N^2\times N^2)$-matrix. The functions $\sigma_{\alpha\beta}$ appearing here are assumed to be continuous bounded functions on $\overline{\Strip(0,\pi)}\to\Cl$ which are analytic in $\Strip(0,\pi)$.

Because of this analyticity, it is clear that $\te\mapsto S(\te)$ has the analytic properties required in Definition~\ref{definition:SMatrix}, and because of the diagonal form \eqref{eq:DiagonalS}, $S$ satisfies items \refitem{item:YangBaxterEqn}, \refitem{item:TCPInvariance} and \refitem{item:TranslationInvariance} of that definition without further constraints on the $\sigma_{\alpha\beta}$. To implement unitarity, hermitian analyticity and crossing symmetry of $S$, one has to require, $\te\in\Rl$, $\alpha,\beta\in\{1,...,N\}$,
\begin{align}
	\overline{\sigma_{\alpha\beta}(\te)}
	&=
	\sigma_{\alpha\beta}(\te)^{-1}
	=
	\sigma_{\beta\alpha}(-\te)
	=
	\sigma_{\beta\alpha}(i\pi+\te)
	\,.
\end{align}
With these constraints on $\sigma_{\alpha \beta}$, it is easy to verify that $S$ as defined in \eqref{eq:DiagonalS} complies with all requirements of Definition~\ref{definition:SMatrix}. We don't give the most general form of the functions $\sigma_{\alpha\beta}$ here, but content ourselves with pointing out that particular examples arise when $\sigma_{\alpha\beta}=\sigma_{\beta\alpha}$ are scalar $S$-matrices of the form \eqref{eq:ScalarS}.
\\
\\
A class of more involved $S$-matrices is given by the scattering operators of $O(N)$ $\sigma$-models, $N\geq3$, see \cite{AbdallaAbdallaRothe:1991, Ketov:2000} for general literature and \cite{BabujianFoersterKarowski:2012} for recent results on the formfactors of these models. These models are defined by quantization of a field theory of $N$ scalar fields $\varphi_1,...,\varphi_N$ whose dynamics is governed by the interaction-free Lagrangian in the presence of the spherical constraint $\sum_{k=1}^N\varphi_k(x)^2=1$. This constraint gives rise to a non-linear field equation, and it turns out that the corresponding field theory is perturbatively renormalizable in $d=1+1$ dimensions and exhibits an infinite number of conservation laws. By making an ansatz exploiting the $O(N)$-symmetry, the factorizing S-matrix of this model can been determined \cite{ZamolodchikovZamolodchikov:1978} (see also further references in \cite{AbdallaAbdallaRothe:1991}, and \cite{ShankarWitten:1978} for a supersymmetric extension).

There has been a lot of interest in non-linear $\sigma$-models in two dimensions because of their similarities to non-Abelian gauge theories in four dimensions, in particular regarding their geometric nature, asymptotic freedom, and instanton solutions. We want to show here that these models fit precisely into the present framework of inverse scattering theory, and thus define them in terms of their $S$-matrix.

In our setting, the $O(N)$ $\sigma$-models can be described as follows. The particle spectrum consists of a single species of neutral particles of mass $m>0$ with an internal degree of freedom transforming under $G=O(N)$, which acts on $\K:=\Cl^N$ by its defining self-conjugate irreducible representation, i.e. in particular $\overline{\alpha}=\alpha$, $\alpha=1,...,N$. The $S$-matrix is defined as \cite[Chapter~8.3.2]{AbdallaAbdallaRothe:1991}
\begin{align}\label{eq:SSigmaModel}
	S_{\sigma,N}(\te)^{\alpha_1\alpha_2}_{\beta_1\beta_2}
	:=
	\sigma_1(\theta)\delta^{\alpha_1\alpha_2}\delta^{\beta_1\beta_2}
	+
	\sigma_2(\theta)\delta^{\alpha_1\beta_2}\delta^{\alpha_2\beta_1}
	+
	\sigma_3(\theta)\delta^{\alpha_1\beta_1}\delta^{\alpha_2\beta_2}
	\,,
\end{align}
with the functions
\begin{align}
	\sigma_2(\theta)
	&:=
	Q(\theta) Q(i\pi-\theta),
	\quad \text{with} \quad
	Q(\theta)
	:=
	\frac{\Gamma(\frac{1}{N-2}-i\frac{\theta}{2\pi})\Gamma(\frac{1}{2}-i\frac{\theta}{2\pi})}
	{\Gamma(\frac{1}{2}+\frac{1}{N-2}-i\frac{\theta}{2\pi})\Gamma(-i\frac{\theta}{2\pi})}\label{eq:sigma2}
	  ,\\
	  \sigma_1(\theta)
	  &:=
	  -\frac{2\pi i}{(N-2)}\,\frac{\sigma_2(\theta)}{i\pi-\theta}
	  \label{eq:sigma1},
	  \\
	  \sigma_3(\theta)
	  &:=
	  \sigma_1(i\pi-\theta)= -\frac{2\pi i}{(N-2)}\,\frac{\sigma_2(\theta)}{\theta}
	  \label{eq:sigma3}.
\end{align}
\begin{proposition}
	The $S$-matrix $S_{\sigma,N}$ defined in \eqref{eq:SSigmaModel}--\eqref{eq:sigma3} complies with Definition~\ref{definition:SMatrix} for the particle spectrum given by $G=O(N)$, $V_1={\rm id}$, $m>0$.
\end{proposition}
\begin{proof}
	As the calculations necessary here can mostly be found in the literature -- see for example \cite{ZamolodchikovZamolodchikov:1979}, and \cite{Schutzenhofer:2011} for a more detailed account --  we will be brief about the proof. To begin with, we note that the function $Q$ is continuous and bounded on $\overline{\Strip(0,\pi)}$, and analytic in the interior of this strip. The analyticity can be checked by verifying that all poles of the Gamma function lie outside the strip, and the boundedness can be established by standard estimates on the Gamma functions \cite{Schutzenhofer:2011}. As $\Gamma$ has a simple pole at the origin, $\te\mapsto Q(\te)$ has a simple zero at $\te=0$. Hence the poles in the fractions appearing in the definitions of $\sigma_1$ and $\sigma_3$ are compensated by zeros, and we conclude that $\sigma_1$, $\sigma_2$, and $\sigma_3$ are analytic in $\Strip(0,\pi)$, and continuous and bounded on the closure of this strip. This implies that $S_{\sigma,N}$ has the analytic properties required in
Definition~\ref{definition:SMatrix}.

	The verification of unitarity \refitem{item:Unitarity}, hermitian analyticity \refitem{item:HermitianAnalyticity} and the Yang-Baxter equation \refitem{item:YangBaxterEqn} can be found in \cite{ZamolodchikovZamolodchikov:1979, Schutzenhofer:2011}. The TCP invariance \refitem{item:TCPInvariance} of $S_{\sigma,N}$ holds because $\overline{\alpha}=\alpha$ for all $\alpha\in\{1,...,N\}$ and $S_{\sigma,N}$ \eqref{eq:SSigmaModel} has the two symmetries $S_{\sigma,N}(\te)^{\alpha_1\alpha_2}_{\beta_1\beta_2}=S_{\sigma,N}(\te)^{\alpha_2\alpha_1}_{\beta_2\beta_1}=S_{\sigma,N}(\te)_{\alpha_1\alpha_2}^{\beta_1\beta_2}$ (corresponding to invariance under parity and time reversal).

	Concerning crossing symmetry, we note that $\sigma_2(i\pi-\te)=\sigma_2(\te)$, $\sigma_1(i\pi-\te)=\sigma_3(\te)$, and $\sigma_3(i\pi-\te)=\sigma_1(\te)$. Hence
	\begin{align*}
		S_{\sigma,N}(i\pi-\te)^{\alpha_1\alpha_2}_{\beta_1\beta_2}
		&=
		\sigma_3(\theta)\delta^{\alpha_1\alpha_2}\delta^{\beta_1\beta_2}
		+
		\sigma_2(\theta)\delta^{\alpha_1\beta_2}\delta^{\alpha_2\beta_1}
		+
		\sigma_1(\theta)\delta^{\alpha_1\beta_1}\delta^{\alpha_2\beta_2}
		=
		S_{\sigma,N}(\te)^{\beta_1\alpha_1}_{\beta_2\alpha_2}
		\,,
	\end{align*}
	i.e., $S_{\sigma,N}$ is crossing symmetric.

	Finally, the mass condition \refitem{item:TranslationInvariance} is trivially satisfied here since only a single mass value appears in the spectrum. By straightforward computation, one also checks that each of the three terms in \eqref{eq:SSigmaModel} is $O(N)$-symmetric in the sense that it commutes with $M\ot M$ for any $M\in O(N)$. Hence also property \refitem{item:GaugeInvariance} holds, and the proof is finished.
\end{proof}

\section{Conclusions}\label{section:conclusions}

Whereas the models treated in \cite{Lechner:2008} were restricted to just one species of neutral particles, the extension carried out here shows that the presented method is also capable of realizing integrable models with any number of particle species, transforming under an arbitrary global gauge group. A particularly interesting class of models which is now accessible by operator-algebraic methods are the nonlinear $O(N)$ $\sigma$-models, which share some features with non-Abelian gauge theories in four dimensions. In view of the thorough analysis these models have seen in other approaches, there seems to be no real doubt that these models do indeed exist as well-defined quantum field theories. By defining these models via their scattering matrix, also a hard existence proof is now within reach in the approach taken here: All that remains to do is to verify the modular nuclearity condition for the $S$-matrix \eqref{eq:SSigmaModel}. This analysis requires quite some technical work and will be presented elsewhere \cite{Alazzawi:2013}.

But already at the present stage good evidence exists which indicates that this condition is likely to hold. The point is that the $\sigma$-model $S$-matrix \eqref{eq:SSigmaModel} does not only comply with Definition~\ref{definition:SMatrix}, but in fact satisfies somewhat stronger regularity properties.
Namely, given any $\eps>0$, the $S$-matrix $\te\mapsto S_{\sigma,N}(\te)$ extends to a bounded and analytic function on an extended strip
$\Strip(-\frac{2\pi}{N-2}+\eps,\pi+\frac{2\pi}{N-2}-\eps)\supset\Strip(0,\pi)$, properly containing the physical region. Furthermore, at $\te=0$ one finds $S_{\sigma,N}(0)=-1$. In the scalar case, the existence of such
a bounded analytic extension on the one hand, leading to sharp Hardy norm estimates, and the value $-1$ of the S-matrix at $\te=0$ on the other hand, improving the nuclearity estimates via the Pauli principle \cite{Lechner:2005}, were essential for establishing the modular nuclearity condition, and thus the existence of local field operators. These mechanisms can probably also be used in the case of the $O(N)$ $\sigma$-models.

It has to be mentioned that a rigorous comparison of models constructed via different techniques like for example continuum limits of lattice theories and inverse scattering theory, respectively, is not straightforward because the quantities that are explicitly accessible depend on the chosen approach. As explained earlier, we take the point of view that the two-particle S-matrix is a good choice for defining the interaction in the case of integrable models. An identification of the models presented here, for example starting from the $O(N)$ $\sigma$-model S-matrix, with $\sigma$-models defined by, say, a Lagrangian and perturbative quantization and renormalization, would best proceed by proving that the solution of the inverse scattering problem does not only exist in this case, but is also unique. Whereas such a proof is currently unmanageable for general quantum field theories, it seems well within reach in the realm of integrable quantum field theories with factorizing S-matrices \cite{Alazzawi:2013}.

One of the most fascinating aspects of these models is their asymptotic freedom, reminiscent of QCD. This property, although often taken for granted, has not been rigorously proven up to now, see \cite{Seiler:2003} for a detailed discussion. Also in the operator-algebraic approach taken here, a proof of asymptotic freedom would require a deeper analysis. However, the basic tools for such an investigation are in place: As in the scaling limits for scalar models \cite{BostelmannLechnerMorsella:2011}, a short distance limit of the $O(N)$ $\sigma$-models decomposes into two chiral massless theories which still contain the $S$-matrix \eqref{eq:SSigmaModel}. Showing that these chiral nets are isomorphic to free field nets would then amount to a proof of asymptotic freedom. We hope to come back to these questions in a future work.


\section*{Acknowledgements} This project has been supported by the FWF project P22929--N16 ``Deformations of quantum field theories''. GL would like to thank H.~Grosse for interesting discussions about $O(N)$ $\sigma$-models during a stay at the Erwin-Schr\"odinger Institute ESI in 2007, and the ESI for its hospitality. Further helpful discussions with D.~Buchholz and W.~Dybalski regarding scattering theory are also gratefully acknowledged.



\end{document}